\documentclass[12pt]{article}

\usepackage{amsmath,amssymb,rotating,wasysym,mathrsfs,xcolor,amsthm}

\newcommand{\bm}[1]{\mbox{\boldmath $#1$}}

\def\be{\begin{equation}}
\def\ee{\end{equation}}
\def\bea{\begin{eqnarray}}
\def\eea{\end{eqnarray}}
\def\bean{\begin{eqnarray*}}
\def\eean{\end{eqnarray*}}

\def\S{\Sigma}

\def\ele{\left(\frac{\Lambda}{3}\right)^{1/2}}

\newcommand{\lied}{\pounds}
\newcommand{\scri}{\mathscr{J}}

\newcommand{\eqs}{\ \stackrel{\scri}{=}\ }

\newcommand{\neqs}{\ \stackrel{\scri}{\neq}\ }

\newcommand{\eqc}{\ \stackrel{\Sc}{=}\ }

\newcommand{\Q}{\mathcal{Q}}
\newcommand{\T}{\mathcal{T}}
\newcommand{\D}{\mathcal{D}}
\newcommand{\W}{\mathcal{W}}
\newcommand{\Z}{\mathcal{Z}}
\newcommand{\Sc}{\mathcal{S}}

\def\nablah{\overline \nabla}

\def\={\stackrel{\scri}{=}}

\def\F{\mathcal F}

\newtheorem{crit}{Criterion}
\newtheorem{Definition}{Definition}
\newtheorem{Remark}{Remark}
\newtheorem{Theorem}{Theorem}
\newtheorem{Proposition}{Proposition}
\newtheorem{Lemma}{Lemma}
\newtheorem{Corollary}{Corollary}

\def\n{\bar{n}}

\begin{document}

\title{Gravitational radiation at infinity with non-negative cosmological constant}

% Authors, for the paper (add full first names)
\author{Jos\'e M. M. Senovilla $^{1,2}$\\
$^{1}$ \quad Departamento de F\'{\i}sica, Universidad del Pa\'{\i}s Vasco UPV/EHU,\\
Apartado 644, 48080 Bilbao, Spain; \\
$^{2}$ \quad EHU Quantum Center , Universidad del Pa\'{\i}s Vasco UPV/EHU;\\
e-mail: josemm.senovilla@ehu.eus}

\maketitle

\begin{abstract}
The existence of gravitational radiation arriving at null infinity $\scri^+$ --i.e. escaping from the physical system-- is addressed in the presence of a non-negative cosmological constant $\Lambda\geq 0$. The case with vanishing $\Lambda$ is well understood and relies on the properties of the News tensor field (or the News function) defined at $\scri^+$. The situation is drastically different when $\Lambda>0$ where there is no known notion of `News' with similar good properties.  In this paper both situations are considered jointly from a {\em tidal} point of view, that is, taking into account the strength (or energy) of the curvature tensors. The fundamental object used for that purposes in the {\em asymptotic (radiant) super-momentum}, a causal vector defined at infinity with remarkable properties. This leads to a novel characterization of gravitational radiation valid for the general case with $\Lambda\geq 0$ that has been proven to be equivalent, when $\Lambda =0$, to the standard one based on News. The implications of this result are analyzed in some detail when $\Lambda >0$. A general procedure to construct `news tensors' when $\Lambda >0$ is depicted, a proposal for asymptotic symmetries provided, and an example of a conserved charge that may detect gravitational radiation at $\scri^+$ exhibited. A series of illustrative examples is listed.
\end{abstract}

\section{Introduction}\label{sec:intro}

The characterization of gravitational radiation escaping (or entering) asymptotically flat spacetimes was firmly established in the 1950-60's \cite{Trautman58,Pirani57,Bel1962,Bondi1962,Sachs1962,Penrose62}, see \cite{Zakharov} and references therein for a comprehensive review of 1973. The covariant approach uses Penrose's conformal completions \cite{Penrose65,Penrose1966,Frauendiener2004,Kroon} and the basic ingredient is the {\em News} tensor field \cite{Bondi1962,Sachs1962}, a tensor that lives at infinity and which, when non-zero, determines univocally the existence of gravitational radiation escaping (or entering) the spacetime. 

Unfortunately, results based on the News tensor apply only to the case with vanishing cosmological constant $\Lambda =0$. Since the beginning of this century we know that the Universe is in accelerated expansion, that proves the existence of a {\em positive} cosmological constant $\Lambda >0$. This constant might be an effective one, or a true new universal constant, but either way it destroys the asymptotically-flat picture, independently of the value of $\Lambda$. Even if $\Lambda$ is minuscule the problem remains. The difficulties were pointed out in \cite{Penrose2011} and largely explained in \cite{Ashtekar2014,Ashtekar2017} where the various problems arising were clearly exposed.

This situation prompted many scientists to attack the problem which resulted in a plethora of results, new techniques, new definitions, and various attempts to recover the neat and nice picture we had when $\Lambda=0$. Nowadays, there is a vast literature on the subject and a better understanding of the predicament when $\Lambda>0$, which can be categorized in the following points
\begin{itemize}
\item Linearized approximations \cite{Ashtekar2015-b}, including  a version of the quadrupole formula in the linear regime \cite{Ashtekar2015,Hoque2019}, the power radiated by a binary system in a de Sitter background \cite{Bonga2017}, or intended definitions of energy \cite{Bishop2016,Kolanowski2020}.
\item Studies using techniques of exact solutions, analyzing the asymptotic behaviour of the Weyl tensor \cite{Krtous2004}, or the radiation generated by accelerating balck holes \cite{Podolsky2009,Griffiths-Podolsky2009}
\item Definitions of mass-energy, by using spinorial techniques \cite{Szabados2013,Szabados2015}, or Newman-Penrose expansions in preferred coordinate systems \cite{Saw2018}, or on null hypersurfaces \cite{Chrusciel2016}, or for weak gravitational waves \cite{Chrusciel2021,Chrusciel2021b}, or using Hamiltonian techniques \cite{Chrusciel2013}, or --for the case of a black hole-- assuming the existence of a timelike Killing vector \cite{Dolan2019}. For a review, see \cite{Szabados2019}.
\item Searching for mass-loss formulas by means of Newman-Penrose formalism using Bondi-type coordinate expansions \cite{Saw2016,Saw2017ziu,Saw2017,Saw2018i,He2016}
\item Using holographic methods, gauge fixing and foliations on $\scri$ , in particular to study asypmtotic symmetries  \cite{Compere2019,Compere2020}  or in combination with Bondi-like coordinate expansions \cite{Poole2019} 
\item Looking for charges and conservation laws \cite{Chrusciel2013,Virmani2019,Kolanowski2021,Poole2021} and references therein.
\item Relation between the radiation and the properties of the sources \cite{He2018}
\end{itemize}

Despite all these advances, a basic problem remained: how to characterize, unambiguously, the presence of gravitational radiation at $\scri$. To solve this funsamental problem, we explored alternative, but physically equivalent, descriptions of the existence of radiation at infinity when $\Lambda =0$.
The main aim in this quest was to find alternatives that could perform equally well in the presence of a positive cosmological constant too. We found an appropriate characterization of gravitational radiation at $\scri$ {\em fully equivalent} to the standard one based on the News tensor \cite{Fernandez-Alvarez_Senovilla20}.
Our proposal was based on a re-scaled version of the {\em Bel-Robinson tensor} \cite{Bel1958,Senovilla97,Bel1962,Senovilla2000} at $\scri$, which describes the tidal energy-momentum of the gravitational field. The News tensor encodes information about {\em quasi-local} energy-momentum radiated away by an isolated system, while the Bel-Robinson tensor describes energy-momentum properties of the {\em tidal} gravitational field ---for historical reasons, one uses the name `super-energy' for this, see Appendix \ref{App:1}. There is a relation between superenergy and quasi-local energy-momentum quantities on closed surfaces \cite{Horowitz82,Senovilla2000,Szabados2004} that can be exploited. Furthermore, actual measurements of gravitational waves are basically of tidal nature.
Hence, it seemed like a good idea to explore the re-scaled Bel-Robinson tensor as a viable object detecting the existence of gravitational radiation. 

Once we had the novel, but equivalent, characterization of radiation we were able to simply use their appropriate version when $\Lambda >0$ and check whether or not it was able to do the job. It certainly is \cite{Fernandez-Alvarez_Senovilla20b}, and we found the fundamental object that can be used for that purposes: the {\em asymptotic (radiant) super-momentum}. This is introduced in section \ref{sec:Q}, where I present our radiation criteria for general $\Lambda\geq 0$. The next section is devoted to clarify the equivalence with the News prescription when $\Lambda =0$, and then section \ref{sec:dS} is devoted to the case with positive $\Lambda$. The problem of the existence of news-like objects in this case, and the question of in- and out-going radiation are discussed in section \ref{sec:news} and the existence of asymptotic symmetries is studied in section \ref{sec:sym}. I end the paper with a list of examples presented in \cite{Fernandez-Alvarez_Senovilla-afs,Fernandez-Alvarez_Senovilla-dS} and some closing comments.

Before that, let us present the set up.

\subsection{Weakly asymptotically simple spacetimes}\label{subsec:prelim}

Throughout, I will assume that the spacetime $(\hat M,\hat g)$ is weakly asymptotically simple admitting a conformal compactification \`a la Penrose \cite{Penrose65,Kroon,Frauendiener2004,Stewart1991}, so that there exists a (unphysical) spacetime $(M,g)$ and a conformal embedding $\Phi : \hat M \hookrightarrow M$ such that 
$$\Phi^* (\Omega^{-2} g) \stackrel{\hat M}{=} \hat g, \hspace{1cm} \Omega\in C^\infty (M), \hspace{1cm} \Omega |_{\Phi(\hat M)} >0$$
where $\Phi^*$ is the pullback of $\Phi$, and that the boundary of the image of $\hat M$ in $M$, denoted by {\color{blue} $\scri :=\partial [\Phi(\hat M)]$}, is a smooth hypersurface where $\Omega$ vanishes:
$$
\Omega \eqs 0, \hspace{1cm}  \bm{n} := d\Omega \neqs 0 .
$$

$\scri$ is called ``null infinity''. When $\Lambda \geq 0$ it consists of two (not necessarily connected) subsets: future ($\scri^+$) and past ($\scri^-$) null infinity, distinguished by the absence of endpoints of past or future causal curves contained in $(M,g)$, respectively.
Under appropriate decaying conditions for the physical Ricci tensor $\hat R_{\mu\nu}$ one has \cite{Penrose65,Kroon}
\be\label{causalcharacter}
n_\mu n^\mu \eqs -\frac{\Lambda}{3} \hspace{4mm} \Longrightarrow \scri \, \, \mbox{is} 
\left\{
\begin{array}{lcr}
\mbox{timelike} & \mbox{if} & \Lambda <0\\
\mbox{null} & \mbox{if} & \Lambda =0\\
\mbox{{\color{blue} spacelike}} & \mbox{if} & \Lambda >0
\end{array}
\right.
\ee
In the cases with $\Lambda \geq 0$, $n_\mu$ is taken to be future pointing.

There is a {\em gauge freedom} by changing the conformal factor by an arbitrary positive factor
\be\label{gauge}
\Omega \rightarrow \Omega\omega, \hspace{1cm} 0< \omega\in C^\infty (M).
\ee
Though this is not necessary, in order to concorde with references \cite{Fernandez-Alvarez_Senovilla20,Fernandez-Alvarez_Senovilla20b,Fernandez-Alvarez_Senovilla-afs,Fernandez-Alvarez_Senovilla-dS} I am going to {\em partly} fix this gauge freedom by choosing $\Omega$ such that $\nabla_\mu \nabla^\mu \Omega \eqs 0$, which in turn implies
\be\label{nablan}
\nabla_\mu n_\nu =\nabla_\mu \nabla_\nu \Omega \eqs 0.
\ee
The remaining gauge freedom is given by functions $\omega>0$ restricted to 
$$
\lied_n \omega =n^\mu\nabla_\mu \omega\eqs 0.
$$

$\scri$ being a hypersurface, it inherits a metric from $(M,g)$, its first fundamental form: 
$$
h(X,Y) := g(X,Y),  \hspace{1cm} \forall X,Y \in \mathfrak{X}(\scri).
$$
Given any basis $\{\vec{e}{}_a\}$ ($a,b,\dots =1,2,3$) of vector fields in $\mathfrak{X}(\scri)$ the corresponding components are denoted by
$$
h_{ab} = g(\vec{e}_a,\vec{e_b}) .
$$
Due to (\ref{causalcharacter}) the metric $h_{ab}$ is Riemannian (positive definite) if $\Lambda >0$, Lorentzian if $\Lambda <0$ and degenerate if $\Lambda =0$. In the latter case, $n^\mu$ is tangent to $\scri$ so that $n^\mu =n^a e^\mu_a$, and then $n^a$ is the degeneration direction
\be
h_{ab} n^a =0, \hspace{1cm} (\Lambda =0).
\ee

For general $\Lambda$, and according to (\ref{nablan}), $\scri$ is a {\em totally geodesic} hypersurface, its second fundamental form vanishing\footnote{In the general case where the partial gauge fixing (\ref{nablan}) is not enforced $\scri$ is a {\em totally umbilic} hypersurface, the second fundamental form being proportional to the first fundamental form.}:
$$
K(X,Y) =0 \hspace{1cm} \forall X,Y \in \mathfrak{X}(\scri) .
$$
This leads to the existence of a canonical torsion-free connection $\overline\nabla$ on $\scri$, inherited from $(M,g)$, {\em independently of the sign of $\Lambda$}:
$$
\overline\nabla_X Y := \nabla_X Y  \hspace{1cm} \forall X,Y \in \mathfrak{X}(\scri).
$$
This connection is, of course, the Levi-Civita connection of $(\scri,h_{ab})$ whenever $\Lambda \neq 0$. Actually, one has
\be\label{nablah}
\overline\nabla_c h_{ab} =0
\ee
for all values of $\Lambda$.

One can also define a volume 3-form $\epsilon_{abc}$ by
$$
- n_\alpha \epsilon_{abc} :\eqs V \eta_{\alpha\mu\nu\rho} e^\mu{}_a e^\nu{}_b e^\rho{}_c .
$$
where $\eta_{\alpha\mu\nu\rho}$ is the canonical volume 4-form in $(M,g)$ and the constant
$$
V=\left\{\begin{array}{ccc}
(|\Lambda|/3)^{1/2} & \mbox{if} & \Lambda \neq 0\\
1 & \mbox{if} & \Lambda =0 .
\end{array}
\right.
$$
Again $\overline{\nabla}_d \epsilon_{abc}=0$ in all cases.

Henceforth, we will say that $S\subset \scri$ is a {\em cut} on $\scri$ if it is a 2-dimensional spacelike submanifold immersed in $\scri$. When $\Lambda >0$ the `spacelike' character is ensured and all possible 2-dimensional submanifolds are cuts. For $\Lambda =0$, cuts are cross sections of the null $\scri$ transversal to the null generators everywhere. In many cases cuts will have $\mathbb{S}^2$ topology, and these always exist in the regular (or asymptotically Minskowskian) case when $\Lambda =0$ as the topology of $\scri$ is $\mathbb{R}\times \mathbb{S}^2$ \cite{Geroch1977}. However, this will not be necessarily the case when $\Lambda >0$ and, furthermore, even in the case with  $\scri \simeq \mathbb{R}\times \mathbb{S}^2$ one might be interested in preferred cuts with non-$\mathbb{S}^2$ topology. Examples are given in \cite{Fernandez-Alvarez_Senovilla-dS}.

%%%%%%%%%%%%%%%%%%%%%%%%%%%%%%%%%%%%%%%%%%

\section{Asymptotic (radiant) super-momentum: the radiation criterion}\label{sec:Q}
The fact is that a real gravitational field is described by the curvature of the spacetime. In particular, gravitational radiation is the propagation of curvature, the propagation of changing geometrical properties, in space and time. Hence, the existence of gravitational radiation carrying energy-momentum --lost by isolated systems in their dynamical evolution-- should be amenable to a description that considers the strength of the curvature, that is, the strength of the tidal gravitational effects, as the fundamental variable. This is the basic idea to be developed in what follows which was put forward and developed in detail in \cite{Fernandez-Alvarez_Senovilla20,Fernandez-Alvarez_Senovilla20b,Fernandez-Alvarez_Senovilla-afs,Fernandez-Alvarez_Senovilla-dS}.

The strength\footnote{This could be called the `energy' of the Weyl curvature, but I prefer to use the word `strength' to avoid misunderstandings, as the physical units are not those of energy \cite{Senovilla97,Senovilla2000}. Actually the name `super-energy' has been traditionally used for these quantities quadratic in the curvature, but this may also lead to confusion. A better name could be the {\em tidal energy}, but we will have to wait to see if this will eventually catch up.} of the tidal gravitational forces can be appropriately described by the {\em Bel-Robinson tensor} (see Appendix \ref{App:1}), defined by
$$
{\cal T}_{\alpha\beta\lambda\mu}=C_{\alpha\rho\lambda}{}^{\sigma}
C_{\mu\sigma\beta}{}^{\rho}+\stackrel{*}{C}_{\alpha\rho\lambda}{}^{\sigma}
\stackrel{*}{C}_{\mu\sigma\beta}{}^{\rho} .
$$
%where $\star C_{\alpha\rho\lambda}{}^{\sigma}:= \frac{1}{2} \eta_{\alpha\rho\mu\nu}C^{\mu\nu}{}_{\lambda}{}^\sigma$.
${\cal T}_{\alpha\beta\lambda\mu}$ is conformally invariant, fully symmetric and traceless
$${\cal T}_{\alpha\beta\lambda\mu}={\cal T}_{(\alpha\beta\lambda\mu)},\hspace{1cm} {\cal T}^{\rho}{}_{\rho\lambda\mu}=0
$$
and satisfies the dominant property
\be\label{DP}
{\cal T}_{\alpha\beta\lambda\mu}u^{\alpha}v^{\beta}w^{\lambda}z^{\mu}\geq 0
\ee
for arbitrary future-pointing vectors $u^{\alpha}$, $v^{\beta}$, $w^{\lambda}$, and $z^{\mu}$ (inequality is strict if all of them are timelike).
The Bel-Robinson tensor is also covariantly conserved 
$$\nabla^{\alpha}{\cal T}_{\alpha\beta\lambda\mu} =0$$
if the $\Lambda$-vacuum Einstein's field equations $R_{\beta\mu}=\Lambda g_{\beta\mu}$ hold.
This provides conserved quantities if there are (conformal) Killing vector fields \cite{Senovilla2000,Lazkoz2003}.
Nevertheless, ${\cal T}_{\alpha\beta\lambda\mu}$ is not a good tensor to describe radiation arriving at {\em infinity}. The reason is that one can prove under very general circumstances that the Weyl tensor vanishes at $\scri$ \cite{Penrose65,Kroon,Geroch1977}: $$C_{\alpha\beta\mu}{}^\nu \eqs 0.$$ Therefore, the Bel-Robinson tensor vanishes there too.

However, the vanishing of the Weyl tensor at $\scri$ allows us to introduce the re-scaled Weyl tensor 
\be\label{def:d}
d_{\alpha\beta\mu}{}^\nu :=\frac{1}{\Omega} C_{\alpha\beta\mu}{}^\nu 
\ee
which is well defined, and generically non-vanishing,  at $\scri$.
This is a conformally invariant traceless tensor field defined on $M$ with the same symmetry and trace properties as the Weyl tensor: it is a Weyl-tensor candidate ---see Appendix \ref{App:1}. In the physical spacetime one has
$$\nabla_{\nu} d_{\alpha\beta\mu}{}^\nu \stackrel{\hat M}{=} \Omega^{-1} \hat\nabla_\nu \hat C_{\alpha\beta\mu}{}^\nu$$
so that $d_{\alpha\beta\mu}{}^\nu$ is divergence-free on $\hat M$ and also at $\scri$ in $\Lambda$-vacuum\footnote{Actually, at $\scri$ it is enough that the physical Cotton tensor decays quickly enough.}.
%if the matter contents decays appropriately at infinity.
The gauge behaviour of the re-scaled Weyl tensor under the remaining gauge freedom (\ref{gauge}) is simply
$$\hspace{1cm} d_{\alpha\beta\mu}{}^\nu \rightarrow \frac{1}{\omega} d_{\alpha\beta\mu}{}^\nu .$$
The Bianchi identities imply that
\be\label{dn}
d_{\alpha\beta\mu}{}^\nu n_\nu +2 \nabla_{[\alpha} S_{\beta]\mu} \eqs 0
\ee
where $S_{\beta\mu} := \frac{1}{2}(R_{\beta\mu} -\frac{1}{6} g_{\beta\mu})$ is the Schouten tensor on $(M,g)$.

%From the vanishing of the Weyl tensor at $\scri$ one deduces the vanishing of the Bel-Robinson tensor there too.
Given that $d_{\alpha\beta\mu}{}^\nu$ is a Weyl-tensor candidate, we can build its super-energy tensor $T\{d\}$ as shown in Appendix \ref{App:1} 
$$
				T\{d\}_{\alpha\beta\gamma\delta} := D_{\alpha\beta\gamma\delta} := \Omega^{-2}\T_{\alpha\beta\gamma\delta}=  d_{\alpha\mu\gamma}{}^\nu d_{\delta\nu\beta}{}^\mu +\stackrel{*}{d}_{\alpha\mu\gamma}{}^\nu \stackrel{*}{d}_{\delta\nu\beta}{}^\mu 
$$
which can also be considered as a \emph{re-scaled Bel-Robinson tensor}.
$\D_{\alpha\beta\gamma\delta}$ is regular at $\scri$, non-vanishing in general.
$\D_{\alpha\beta\gamma\delta}$ has all the properties of the Bel-Robinson tensor, in particular is fully symmetric and traceless. It is also divergence-free at $\scri$ under the decaying conditions for the physical energy-momentum tensor that imply $\nabla_\nu d_{\alpha\mu\gamma}{}^\nu\eqs 0$.
 Its gauge behaviour under (\ref{gauge}) is 
				$$\D_{\alpha\beta\gamma\delta}\rightarrow \frac{1}{\omega^2}\D_{\alpha\beta\gamma\delta}.$$
				
From now on we will concentrate in the physical relevant case with non-negative $\Lambda\geq 0$.
The fundamental object on which the entire approach is based is the following one-form			
\be\label{Pi}
\fbox{$\Pi_\alpha := -n^\mu n^\nu n^\rho \D_{\alpha\mu\nu\rho} =-\nabla^\mu \Omega \nabla^\nu \Omega \nabla^\rho \Omega \D_{\alpha\mu\nu\rho}  $}
\ee
which is geometrically well and uniquely defined at $\scri$. We will mainly use the properties of $\Pi_\alpha$ at $\scri$. From the general dominant property of super-energy tensors (Appendix \ref{App:1}) one knows that $\Pi_\alpha|_\scri$ is causal and future pointing ---this is also true on a neighbourhood of $\scri$ when $\Lambda >0$, and can always be achieved on such a neighbourhood when $\Lambda =0$ by appropriate choices of $\Omega$. In general, we call $\Pi_\alpha|_\scri$ the {\em asymptotic super-momentum}. Actually, in the $\Lambda =0$ situation, $\Pi_\alpha|_\scri$ is null, and to stress this fact we add the adjective ``radiant'' and then a specific notation is used:
\bean
&\Lambda =0:& \hspace{7mm} \Pi_\mu |_\scri := \Q_\mu, \hspace{7mm} \Q_\mu \Q^\mu =0 \hspace{2mm} \mbox{(Asymptotic radiant super-momentum)}\\
&\Lambda >0:& \hspace{7mm} \Pi_\mu |_\scri := p_\mu, \hspace{7mm} p_\mu p^\mu \leq 0 \hspace{2mm} \mbox{(Asymptotic super-momentum)}
\eean

The gauge behaviour under (\ref{gauge}) is the same for both $\Q_\mu$ and $p_\mu$, namely we have in general
$$\Pi_\alpha |_\scri \rightarrow \omega^{-5} \Pi_\alpha |_\scri .
$$
% \left(\Q_\alpha -3\frac{\Omega}{\omega} \D_{\alpha\beta\rho\tau}n^\beta n^\rho \nabla^\tau\omega\right) +O(\Omega^2)$
Furthermore we have the following important property
\be \label{divPi}
					\nabla_\mu\Pi^\mu\eqs 0  
\ee
which holds in full generality when $\Lambda =0$ \cite{Fernandez-Alvarez_Senovilla-afs}, but needs to assume 
that the energy-momentum tensor of the physical space-time $(\hat M,\hat g_{\mu\nu})$ behaves approaching $\scri$ as $ \hat{T}_{\alpha\beta}|_\scri \sim \mathcal{O}(\Omega^3)$ \cite{Fernandez-Alvarez_Senovilla-dS} (this includes the vacuum case $ \hat{T}_{\alpha\beta}=0$).	

The existence of gravitational radiation cannot be detected at a given point, due to the non-local nature of the gravitational field. Thus, the maximum one can aspire for is to detect the radiation by tidal deformations of cuts \cite{Penrose1986}. Consider thus any cut $S\subset \scri$ and let $\ell^\mu$ be a null normal to $S$ such that $\bm{\ell}\wedge \bm{n} \neq 0$. 
The criteria that we found to detect the existence or absence of gravitational radiation arriving at $\scri^+$ (or departing from $\scri^-$) are as follows \cite{Fernandez-Alvarez_Senovilla20,Fernandez-Alvarez_Senovilla20b,Fernandez-Alvarez_Senovilla-afs,Fernandez-Alvarez_Senovilla-dS}
\begin{crit}[Absence of radiation on a cut]\label{crit1}
When $\Lambda \geq 0$, there is no gravitational radiation on a cut $S\subset \scri$ with spherical topology if and only if $\Pi_\alpha|_S$ is orthogonal to $S$ pointing along the direction $\ell_\alpha +\mbox{{\rm sgn}} (\Lambda) \left(n_\alpha -\ell_\alpha\right)$.
\end{crit}		
Observe that this criterion states that $p_\mu$ points along $n_\mu$ if $\Lambda >0$, and that if $\Lambda =0$, $\Q_\mu$ points along $\ell_\mu$ (which in this case is uniquely determined as the null direction orthogonal to $S$ other than $n_\mu$).

The restriction on the topology of the cut will be justified later when we discuss the equivalence with the standard characterization of a vanishing news tensor if $\Lambda =0$. However, such a restriction can be somewhat relaxed if one considers open portions of $\scri$. Thus, 
let now $\Delta\subset \scri$ denote an open portion of $\scri$ with the same topology of $\scri$.
\begin{crit}[Absence of radiation on $\Delta \subset \scri$]\label{crit2}
When $\Lambda \geq 0$, there is no gravitational radiation on an open portion $\Delta\subset \scri$ that admits a cut with $\mathbb{S}^2$-topology if and only if $\Pi_\alpha|_\Delta$ is transversal to $\scri$ and orthogonal to $\Delta$. This is the same as saying that $\Pi_\alpha|_\Delta$ is orthogonal to every cut within $\Delta$.

Equivalently, there is no gravitational radiation on such open portion $\Delta\subset \scri$ if and only if $n_\alpha|_\Delta$ is a principal direction of the re-scaled Weyl tensor $d_{\alpha\beta\lambda\mu}$ there.
\end{crit}		
Observe that these criteria are identical for both cases with positive or zero $\Lambda$, and that they are purely geometrical and fully determined by the algebraic properties of $d_{\alpha\beta\lambda\mu}$. Here, the principal directions of the Weyl-tensor candidate $d_{\alpha\beta\lambda\mu}$ are considered in the classical sense \cite{Pirani57,Bel1962}, that is, those lying in the intersection of the principal planes, or in other words, the common directions of the eigen-2-forms of $d_{\alpha\beta\lambda\mu}$ when seen as an endomorphism on 2-forms. Recall that, considering only the {\em causal} principal vectors, for Petrov type I there is one principal {\em timelike} vector and no null one, for Petrov type D there is an entire 2-plane of causal principal directions --which contains the two null multiple null ones-- and finally for Petrov types II, III, or N there is just one null principal vector and no timelike one.

Let me then make some brief considerations about the implications of these criteria from the viewpoint of the algebraic properties of the re-scaled Weyl tensor. In the case with $\Lambda =0$, stating that $\Q_\alpha$ is orthogonal to $\Delta\subset \scri$ {\em and} transversal to $\scri$ can only happen if $\Q_\alpha$ actually vanishes there $\Q_\alpha|_\Delta =0$. But this is known to imply \cite{Bergqvist98,Senovilla2011} that the null $n^\mu$ is actually a multiple principal null direction of $d_{\alpha\beta\lambda\mu}|_\Delta$, that is to say, the re-scaled Weyl tensor is algebraically special and, at least, of Petrov type II there, which is in accordance with the discussion in \cite{Krtous2004}. Hence, if $d_{\alpha\beta\lambda\mu}$ is type I and $\Lambda =0$, the existence of radiation is ensured. In the case with $\Lambda >0$, $p_\mu$ is orthogonal to $\Delta\subset \scri$ (and then automatically transversal too) if $p_\mu$ points along the normal $n_\mu$, so that $\bm{p}\wedge \bm{n}=0$. This states that the `asymptotic' super-Poynting (see later section \ref{subsec:superp}) relative to the frame defined by $n^\mu$ vanishes, that is
$$
\left(\delta^\mu_\nu -\frac{3}{\Lambda}  n^\mu n_\nu\right)p^\nu \stackrel{\Delta}{=} 0 
$$
which implies that $n^\mu$ is a principal vector of $d_{\alpha\beta\lambda\mu}$ \cite{Bel1962,Ferrando1997}. As $n_\mu$ is timelike in this situation, absence of radiation in this case requires that $d_{\alpha\beta\lambda\mu}|_\Delta$ is of Petrov type I or D. The converse does not hold; for instance, the $C$-metric is Petrov type D and contains gravitational radiation, see section \ref{sec:fin} and \cite{Fernandez-Alvarez_Senovilla-dS}.

There should be no confusion between the Petrov type of the physical Weyl tensor $\hat{C}_{\alpha\beta\lambda}{}^\mu$ and that of $d_{\alpha\beta\lambda}{}^{\mu}$. Of course, there is a relation between them, as the Petrov type of the latter can only be equally, or more, degenerate than that of the former in the asymptotic region. This follows because the Weyl tensor is conformally invariant so that $\hat{C}_{\alpha\beta\lambda}{}^\mu\stackrel{\hat M}{=} C_{\alpha\beta\lambda}{}^\mu$ and therefore, using (\ref{def:d}) the Petrov type of $d_{\alpha\beta\lambda}{}^{\mu}$ is the same as that of $\hat{C}_{\alpha\beta\lambda}{}^\mu$ on a neighbourhood of $\scri$. By using any invariant characterization of the Petrov types, for instance with curvature invariants or the number of principal null directions, one easily deduces that the Petrov type of $d_{\alpha\beta\lambda}{}^\mu$ at $\scri$ is as degenerate, or more, than that of the physical Weyl tensor near $\scri$. The reasoning is that, if one of the invariants used in the classification \cite{Stephani2003} vanishes in the neighbourhood of $\scri$ then it will also vanish at $\scri$, while if it does not vanish on the neighbourhood, it may vanish or not at $\scri$. Therefore, the possible Petrov types of $d_{\alpha\beta\lambda}{}^{\mu}$ are restricted as follows
\begin{itemize}
\item If the Petrov type of $\hat{C}_{\alpha\beta\lambda}{}^{\mu}$ in the asymptotic region is I, then $d_{\alpha\beta\lambda}{}^{\mu}$ can have {\em any} Petrov type at $\scri$
\item If the Petrov type of $\hat{C}_{\alpha\beta\lambda}{}^{\mu}$ in the asymptotic region is II, then $d_{\alpha\beta\lambda}{}^{\mu}$ can have any Petrov type at $\scri$ except I
\item If the Petrov type of $\hat{C}_{\alpha\beta\lambda}{}^{\mu}$ in the asymptotic region is III, then $d_{\alpha\beta\lambda}{}^{\mu}$ can have Petrov types III, N and 0 at $\scri$
\item If the Petrov type of $\hat{C}_{\alpha\beta\lambda}{}^{\mu}$ in the asymptotic region is N, then $d_{\alpha\beta\lambda}{}^{\mu}$ is either Petrov type N or 0 at $\scri$
\item If the Petrov type of $\hat{C}_{\alpha\beta\lambda}{}^{\mu}$ in the asymptotic region is D, then $d_{\alpha\beta\lambda}{}^{\mu}$ is either Petrov type D or 0 at $\scri$
\item If $\hat{C}_{\alpha\beta\lambda}{}^{\mu}=0$ on an open asymptotic region, then $d_{\alpha\beta\lambda}{}^{\mu}\eqs 0$
\end{itemize}
Hence, all Petrov types on the asymptotic region of the physical spacetime --except 0-- are compatible with the existence, and with the absence, of gravitational radiation crossing $\scri$.

In what follows, first I will show that criterion \ref{crit2} coincides with the traditional one when  $\Lambda =0$, and then I will discuss the implications that it has when $\Lambda >0$.

\section{The case with $\Lambda =0$: equivalence with the news criterion}\label{sec:Lambda=0}
As we saw in subsection \ref{subsec:prelim}, if $\Lambda =0$ $n_\mu $ is null, $h_{ab}$ is degenerate, $n^\mu \eqs n^a e^\mu{}_a$ and $n^a$ is the degeneration vector field at $\scri$, ergo tangent to its null generators: $h_{ab}n^a =0$. Using the canonical connection and (\ref{nablan}), $n^a$ is parallel on $\scri$: 
\be\label{nablan1}
\overline\nabla_b n^a =0.
\ee

The topology of $\scri$ is usually taken to be $\mathbb{R}\times \mathbb{S}^2$, though there are cases where this does not hold if there are singularities or incompleteness of $\scri$. In the standard case with $\scri\simeq \mathbb{R}\times \mathbb{S}^2$, the cuts $\Sc$ can be chosen to be topologically $\mathbb{S}^2$, see figure \ref{fig:cut}.
For any cut $\Sc$ there is a {\em unique} lightlike vector field $\ell^\mu$ orthogonal to $\Sc$ and such that $n_\mu \ell^\mu =-2$ ---this is the vector field $\ell^\mu$ used in criterion \ref{crit1}. I will denote by $\{\vec{E}_A\}$ any basis of $\mathfrak{X}(\Sc)$ ($A,B,\dots =2,3$). These can be extended to vector fields on $\scri$ by choosing them on any cut and then propagating them such that $\lied_n E^a_A = M_A n^a$ (for some $M_A$ which will be irrelevant in what follows), where $\lied_v$ is the Lie derivative with respect to $v^a$ on $\scri$. Then $\{\vec{e}_a\}=\{\vec n,\vec{E}_A\}$ are a basis of vector fields on $\scri$. Let $h^{ab}$ represent any tensor field satisfying 
$$h^{ab} h_{ac} h_{bd} = h_{cd}.$$
Such $h^{ab}$ suffers from an indeterminacy as $h^{ab} +n^a s^b + n^b s^a$ also satisfies the condition, for arbitrary $s^b$. Nevertheless, $h^{ab}$  allows us to raise indices and take traces {\em unambiguously} when acting on covariant tensors fully orthogonal to $n^a$.

\begin{figure}[h]
\includegraphics[height=10cm]{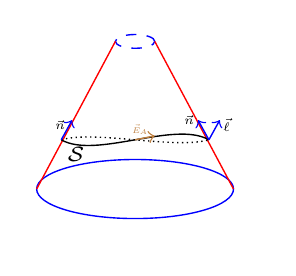}
\caption{This is a schematic representation of $\scri^+$ when $\Lambda =0$, where $\vec n$ is the null degeneration vector field, $\Sc$ is a cut, $\vec\ell$ is the unique null vector orthogonal to $\Sc$ and transversal to $\scri$, and $\vec E_A$ are vector fields tangent to the cut. Cuts are 2-dimensional surfaces, usually with $\mathbb{S}^2$ topology. In the picture, one dimension is suppressed, so that here this topology of the cut is represented as a circumference. \label{fig:cut}}
\end{figure}

The connection $\overline\nabla$, which is inherited from the spacetime, has a curvature tensor $\overline{R}_{abc}{}^d$ and the corresponding (symmetric) Ricci tensor $\overline{R}_{ac}:=\overline{R}_{adc}{}^d$. It happens that
$$
\overline{R}_{ab} n^b =0
$$
and therefore 
\be\label{Rbar}
\overline{R} := h^{ab} \overline{R}_{ab} 
\ee
is well defined.

Due to (\ref{nablah}) and to the vanishing of the second fundamental form on $\scri$, which induces (\ref{nablan1}), in this case we also have on $\scri$
$$\lied_n h_{ab}=0.$$
Hence, {\em all possible cuts are isometric}, with a first fundamental form 
$$
q_{AB}:=h_{ab} E^a_A E^b_B, \hspace{15mm} \lied_n q_{AB} =n^c\nablah_c q_{AB} =0
$$
which is basically the non-degenerate part of $h_{ab}$. Its covariant derivative will be denoted by $D_A$. The scalar curvature (or twice the Gaussian curvature) of the cuts is precisely (\ref{Rbar}) ---and $\lied_n \overline R =0$. Of course, only the conformal class is fixed because of the gauge freedom (\ref{gauge}): 
\be\label{gaugeh}
h_{ab} \rightarrow \tilde{h}_{ab}\eqs \omega^2 h_{ab},\hspace{1cm}   \tilde{q}_{AB} \eqc \omega^2 q_{AB}.
\ee

The structure $(h_{ab},n^a)$ on $\scri$ is universal. Nevertheless, observe that it does not contain any {\em dynamical} behaviour. The dynamics, and therefore the possible existence of gravitational radiation, is not encoded in this universal structure: it comes from structure {\em inherited} from the physical spacetime. In this $\Lambda =0$ situation the time dependence along $\scri$ is actually encoded in the connection $\overline\nabla$ and its curvature. This is crucial. Notice that
$$
\lied_n \overline\nabla \neq 0, \hspace{1cm} [\lied_n,\overline\nabla]\neq 0
$$
In particular, for any one-form $\bm{t}$ 
\be
[\lied_n,\overline\nabla_b]t_a = -n^c t_c \, \left(\overline{S}_{ab}-\frac{1}{2}h^{ef}\overline{S}_{ef} h_{ab}\right)
\ee
where $\overline{S}_{ab}$ is the pull-back of the Schouten tensor to $\scri$:
$$\overline{S}_{ab}:\eqs S_{\mu\nu}e^\mu{}_a e^\nu{}_b , \hspace{1cm} n^a \overline{S}_{ab} =0$$
also given by
$$\overline{S}_{ab}-\frac{1}{2}h^{ef}\overline{S}_{ef} h_{ab}= \overline{R}_{ab} -\frac{1}{2}\overline R h_{ab}.$$

In plain words, $\overline{S}_{ab}$ encodes the time variations within $\scri$, hence it contains the information about any gravitational radiation crossing $\scri$.
However, $\overline{S}_{ab}$ has non-trivial gauge behaviour:
\be\label{gaugeS}
\overline{S}_{ab} \rightarrow \overline{S}_{ab} - \frac{1}{\omega}\overline\nabla_a \overline\nabla_b \omega + \frac{2}{\omega^2}\overline\nabla_a\omega\overline\nabla_b\omega-\frac{1}{2\omega^2} h_{ab}\, \omega^c\overline\nabla_c \omega
\ee
(here $g^{\mu\nu}\nabla_\nu\omega : \eqs \omega^c e^\mu{}_c$).
One needs to extract the relevant gauge-invariant part of $\overline{S}_{ab}$: this is the News tensor field.

There are many ways to define the News tensor field, such as by using expansions in Bondi coordinates \cite{Bondi1960,Bondi1962,Sachs1962}, or defining the asymptotic outgoing shear \cite{Ashtekar81,Kroon,Penrose65,Stewart1991}, or by computing the limit at $\scri$ of $\Omega^{-1} \nabla_\mu n_\nu$ in certain gauges \cite{Winicour}.
To our purposes, the best suited definition is just the dynamical (time-dependent) and gauge invariant part of $\overline{S}_{ab}$, in accordance with \cite{Geroch1977}. This is a geometrically neat and physically clarifying definition. 

To find the explicit expression, I start by noticing that $\overline{S}_{ab}$ is orthogonal to $n^a$, so that only the components $S_{AB} =\overline{S}_{ab}E^a_A E^b_B$ are non-zero. Nevertheless, these components change from cut to cut, due to the dynamical dependence of $\overline{S}_{ab}$ itself. By projecting \eqref{dn} to $\scri$ one has
\be\label{bianchi}
2\nablah_{[a}\overline{S}_{b]c}=-e^\alpha_a e^\beta_b e^\lambda_c d_{\alpha\beta\lambda}{}^\mu n_\mu
\ee
from where it easily follows
$$
\lied_n \overline{S}_{bc} = n^c\nablah_c \overline{S}_{bc} = -n^\alpha e^\beta_b e^\lambda_c d_{\alpha\beta\lambda}{}^\mu n_\mu\neq 0
$$
which is non-vanishing in general. In particular
$$
\lied_n S_{AB} = n^c\nablah_c S_{AB} \neq 0
$$
so that $S_{AB}$ depend on the cut. Such a time-dependent part is what interests us. Consequently, we need tu subtract, from $\overline{S}_{ab}$, a tensor field that is symmetric, orthogonal to $n^a$, time-independent and with a gauge behaviour that compensates \eqref{gaugeS}. Explicitly, we need a tensor field $\rho_{ab}$ such that
\be\label{rhoprop}
\rho_{ab}=\rho_{ba}, \hspace{1cm} n^a\rho_{ab}=0, \hspace{8mm} \overline\nabla_{[c}\rho_{a]b} =0,
\ee
and with the following gauge behaviour under \eqref{gaugeh}
$$
\tilde{\rho}_{ab}= \rho_{ab} - \frac{1}{\omega}\overline\nabla_a \overline\nabla_b \omega + \frac{2}{\omega^2}\overline\nabla_a\omega\overline\nabla_b\omega-\frac{1}{2\omega^2} h_{ab}\, \omega^c\overline\nabla_c \omega .
$$
Note that $n^c\nablah_c \rho_{ab}=0$ follows from the above, so that $\rho_{ab}$ is actually a true 2-dimensional tensor field, with only $\rho_{AB}$ non-zero components and these are time-independent $n^c\nablah_c \rho_{AB}=0$. Therefore, it is enough to have this tensor field on any cut. {\em But this is the tensor $\rho_{AB}$ studied in Appendix \ref{App:rho}}. Observe that then we have, in addition, $h^{ab} \rho_{ab} =\overline{R}/2$.

The News tensor field is defined by \cite{Geroch1977}
\be\label{news}
\fbox{$N_{ab}:= \overline{S}_{ab} -\rho_{ab}$}
\ee
and has the following properties
$$
N_{ab}=N_{ba}, \hspace{1cm} n^a N_{ab} =0, \hspace{1cm} h^{ab} N_{ab} =0
$$
and, more importantly, $N_{ab}$ {\em is gauge invariant} under \eqref{gaugeh}
$$
N_{ab}=\tilde{N}_{ab}.
$$
From \eqref{bianchi}, \eqref{news} and \eqref{rhoprop} we derive
\be\label{bianchiN}
2\nablah_{[a}N_{b]c}=-e^\alpha_a e^\beta_b e^\lambda_c d_{\alpha\beta\lambda}{}^\mu n_\mu
\ee
from where, as before,
$$
\lied_n N_{ab} \neq 0
$$
in general, so that the News tensor generically changes from one cut to another.
The pullback of $N_{ab}$ to any cut $\Sc$ is denoted by 
$$
N_{AB}(\Sc)\eqc N_{ab} E^a{}_A E^b{}_B.
$$
I will also use the notation
$$
\dot N_{AB}(\Sc) :\eqc E^a{}_A E^b{}_B \lied_n N_{ab}
$$
The classical characterization of gravitational radiation in the case $\Lambda =0$ is given as follows
\begin{Definition}[Classical radiation characterization]\label{def:news}
There is no gravitational radiation on a given cut $\Sc\subset \scri$ if and only if the News tensor vanishes there:
$$
N_{AB}(\Sc)=0 \Longleftrightarrow N_{ab}\eqc 0 \Longleftrightarrow \mbox{no gravitational radiation on $\Sc$}
$$
\end{Definition}
\begin{Remark}
Observe that $N_{ab}$ is a tensor field, and its vanishing at any point is an invariant statement. Nevertheless, one cannot aspire to localize gravitational radiation at a point, and thus the vanishing of $N_{ab}$ at a given point has no meaning in principle --see e.g. the discussion in \cite{Penrose1986}. On the other hand, the vanishing of $N_{ab}$ on an entire cut does have a meaning, as this is a quasilocal statement. In this sense, $N_{ab}$ is related to the quasi-local energy-momentum properties of the gravitational field at $\scri$.
\end{Remark}

To justify the previous definition, a description of the gravitational energy-momentum properties at infinity is needed, which in turn requires the knowledge of the asymptotic symmetries, that is, the symmetries of $\scri$: the BMS group \cite{Geroch1977,Bondi1962,Winicour,Penrose1966,Sachs62}. A convenient characterization of the infinitesimal isometries of $\scri$ that is independent of the gauge choice is given by the vector fields $\vec Y\in \mathfrak{X}(\scri)$ satisfying
$$
\lied_Y (n^a n^b h_{cd})=0.
$$
This can be shown to be equivalent to ($\phi\in C^\infty(\scri)$)
$$
\lied_Y n^b =-\phi n^b, \hspace{1cm} \lied_Y h_{ab} = 2\phi h_{ab}
$$
and the set of such vector fields is a Lie algebra.
Any vector field of the form $Y^a=\alpha n^a$, with $\lied_n \alpha =0$ (and gauge behaviour $\tilde\alpha =\omega\alpha$), satisfies these relations. These are called {\em infinitesimal super-translations}, and constitute an infinite-dimensional Abelian ideal. The rest of the BMS algebra is given by the conformal Killing vectors of $(\Sc,q_{AB})$ (i.e., the Lorentz group for round spheres).
There exists, however, a 4-dimensional Abelian sub-ideal constituted by the solutions of the linear equation ($\Delta$ is the Laplacian on $(\Sc,q_{AB})$, see Appendix \ref{App:rho})
$$
\overline\nabla_a \overline\nabla_b \alpha +\alpha \rho_{ab}-\frac{1}{2}h_{ab} \left(\Delta \alpha + \frac{\overline R}{2} \right)=0 
$$
whose elements are called {\em infinitesimal translations}. This equation is fully orthogonal to $n^a$ and time independent (its Lie derivative with respect to $n^a$ vanishes), and thus it is actually fully equivalent to the equation on any given cut
$$
D_A D_B \alpha -\frac{1}{2}q_{AB} \Delta \alpha +\alpha \left(\rho_{AB} -\frac{\overline R}{4}q_{AB} \right)=0 .
$$
{\em This is precisely equation \eqref{eq:Hess-gen}} whose four independent solutions are denoted by $\pi_{(\mu)}$.
Using these solutions $Y^b_{(\mu)} :=\pi_{(\mu)} n^b$, the corresponding {\em Bondi-Trautman 4-momentum} on any given cut $\Sc$ can be expressed as \cite{Geroch1977}
$$
B_{(\mu)} (\Sc) :=-\frac{1}{32\pi}  \int_\Sc \pi_{(\mu)} \left( d_{\beta\mu\nu}{}^\rho n_\rho\ell^\beta n^\mu\ell^\nu +2\sigma^{AB}N_{AB}\right) %\bm{\epsilon}_2 
$$
where $\sigma_{AB}$ is the shear tensor of $\Sc$ along $\ell^\mu$, that is to say, the trace-free part of $E^\mu{}_A E^\nu{}_B \nabla_\mu\ell_\nu$ on $\Sc$.
%The positive solution $\alpha >0$ corresponds to the {\em timelike} translation.
\begin{figure}
\includegraphics[height=10cm]{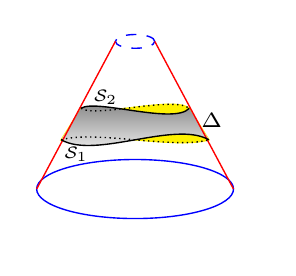}
\caption{Schematic representation of a portion $\Delta$ of $\scri^+$ delimited by two cuts $\Sc_1$ and $\Sc_2$ when $\Lambda =0$, one dimension suppressed. The cut $\Sc_2$ is to the future of $\Sc_1$. The portion $\Delta$ has the same topology as $\scri^+$ and is depicted by the shadowed part.\label{fig:Delta}}
\end{figure}
Let now $\Delta\subset\scri^+$ be a connected open portion of $\scri^+$, with the same topology as $\scri^+$, and limited by two cuts $\Sc_1$ and $\Sc_2$, with $\Sc_2$ entirely to the future of $\Sc_1$, as shown in figure \ref{fig:Delta}. 
One can compute the Bondi-Trautman 4-momentum for both cuts, and check what is the difference. 
The result is, removing any matter content around $\scri^+$ for simplicity and to make things clearer (for the general case see e.g. \cite{Geroch1977,Fernandez-Alvarez_Senovilla-afs}), 
$$
B_{(\mu)}(\Sc_2)-B_{(\mu)}(\Sc_1) = -\frac{1}{32\pi}\int_{\Delta} \pi_{(\mu)}  h^{ab} h^{cd} N_{ac} N_{bd} 
$$
which is a null vector in the auxiliary Minkowski metric of Appendix \ref{App:rho} where $\eta^{\mu\nu} \pi_{(\mu)}\pi_{(\mu)}=0$ and in particular has a strictly negative 0-component. This leads to the interpretation of News of definition \ref{def:news}.

We can finally prove the equivalence of definition \ref{def:news} with Criteria \ref{crit1} and \ref{crit2}. 
On a given cut $ \Sc $, one can split the radiant super-momentum into its null transverse (along $\ell^\alpha$) and tangent parts to $ \scri $,
$$
			\Q^\alpha\eqc %\W \ell^\alpha + \overline{\Q}^\alpha=
			\frac{1}{2}\W \ell^\alpha + \overline{\Q}^a e^\alpha{}_a,$$
		where
$\W := -n^\mu \Q_\mu \geq 0$ and 
$$
			\overline{\Q}^a :=\frac{1}{2} \Z n^a + \overline{\Q}^A E^a{}_A\quad \text{with}
			\hspace{3mm} \Z := -\ell_\mu \Q^\mu \geq 0 .
$$
These quantities are observer-independent: $\Z$ and $\overline{\Q}^A$ depend only on the cut, while $\W$ is fully intrinsic to $\scri$. 

The theorem that proves equivalence with criterion \ref{crit1} is:
\begin{Theorem}[Radiation condition]\label{th}
				 There is no gravitational radiation on a given cut $ \Sc \subset \scri $ with $\mathbb{S}^2$ topology if and only $\Q^\mu$ points along $\ell^\mu$ on that cut:
$$
				   N_{AB}(\Sc)= 0 \quad\Longleftrightarrow \quad \overline{\Q}^a\eqc 0 \quad (\Longleftrightarrow \quad \Z = 0).
$$
				
\end{Theorem}
\begin{proof}
Projecting \eqref{bianchiN} to $\Sc$, a somehow long calculation leads to
\begin{align}
			\W&\eqc 2\dot {N}^{RT}\dot {N}_{RT} \geq 0 ,\label{nQ}\\
			\Z&\eqc 8 D^{[A} N^{B]C} D_{[A} N_{B]C}=4 D_C N^C{}_A D_B N^{BA}  \geq 0 ,\label{lQ}\\
			\overline{\Q}^A &\eqc 	 8\dot N_{BC} D^{[B} N^{A]C} = -4 \dot N^{BA} D_C N^C{}_B .\label{FQ}	
			\end{align}
Eq. \eqref{lQ} implies that $\Z= 0 \iff D_{[A} N_{B]C} =  0 $. Using \eqref{FQ}, this happens if and only if $ \overline\Q^a= 0 $, that is, if and only if $2\Q^\mu \eqc \W \ell^\mu$.
But $D_{[A} N_{B]C} =  0$ ---or equivalently $D_A N^A{}_B =0$--- informs us that $N_{AB}$ is a traceless symmetric Codazzi (and divergence-free) tensor on the compact $\Sc$, which implies \cite{Liu1998} that $N_{AB}=0$. Hence $  N_{AB}= 0  \iff \overline{\Q}^a= 0 $ on $\Sc$ .
			\end{proof}
			
\begin{Remark}
As the radiant super-momentum $\Q^\mu$ is always null, this theorem can be equivalently stated as: there is no gravitational radiation on a given cut $ \Sc \subset \scri $ if and only if the radiant super-momentum is orthogonal to $\Sc$ everywhere and not co-linear with $n^\alpha$. 
Notice that, given a cut, this statement is totally unambiguous.			
%That is to say: $N_{AB}(\Sc)= 0 \iff  \Q^\alpha\eqc \W\ell^\alpha$. , with $  \W\neqc 0$  in general.
\end{Remark}

Similarly, the theorem that proves equivalence with criterion \ref{crit2} is:
\begin{Theorem}[No radiation on $ \Delta\subset\scri$] 
There is no gravitational radiation on an open portion $ \Delta \subset \scri $ which contains a cut with topology $\mathbb{S}^2$ if and only if the radiant super-momentum $ \Q{^\alpha} $ vanishes on $\Delta$:
$$
				   N_{ab}\stackrel{\Delta}{=}  0 \quad\Longleftrightarrow \quad \Q{^\alpha}\stackrel{\Delta}{=}  0 .
$$
\end{Theorem}	
\begin{proof}
If one can find cuts with $\mathbb{S}^2$ topology in $\Delta$, then according to the previous remark and theorem \ref{th}, absence of radiation on $\Delta$ requires that $2\Q^\alpha\eqc \W \ell^\alpha$ on {\em every possible such cut} $\Sc$ included in $\Delta$. But this is only possible if $\Q^\alpha\stackrel{\Delta}{=} 0$. 
More generally, observe first that $N_{ab}\stackrel{\Delta}{=}  0$ trivially implies $\Q^\alpha\stackrel{\Delta}{=}  0$ due to \eqref{nQ}-\eqref{FQ} independently of the topologies. Conversely, if $\Q^\alpha \stackrel{\Delta}{=}  0$, then from \eqref{nQ} $\dot N_{AB}\stackrel{\Delta}{=}  0$, so that $N_{ab}$ is time independent and $N_{AB}$ is the same for all possible cuts (as they are all locally isometric). From \eqref{lQ} we also have $D_{[A} N_{B]C}=0$ on every cut. Thus, if a compact cut has a positive Gaussian curvature --so that its topology is necessarily $\mathbb{S}^2$, then a known theorem \cite{Liu1998} implies that $N_{AB}=0$.
\end{proof}

\begin{Remark}
If there is gravitational radiation at $\scri$, there can arise situations where actually $2\Q^\mu =\W\ell^\mu\neq 0$ for a given foliation of cuts, with $\Z= 0$ on them. Of course, this is only possible if the cuts have a non-$\mathbb{S}^2$ topology. In this case, on those cuts $D_{[A} N_{B]C}=0$ (and $D_B N^{BA}=0$). In particular, for instance if $\overline R =0$ one further has $D_C N_{AB}=0$, so that $N_{AB}$ is constant on those cuts. Hence, $N_{ab}=N_{ab}(v)$ are functions of a single coordinate $v$ such that the foliation is defined by $v=$const.\, and necessarily $n^a\nablah_a v \neq 0$. For any other cut not in this special foliation $\Z\neq 0$. In any case, the non-vanishing of $\Q^\mu$ detects the radiation in this case correctly. Some examples of this situation exist in the C-metric and the Robinson-Trautman solutions.
\end{Remark}

\section{The case with $\Lambda >0$}\label{sec:dS}
The case of asymptotically de Sitter spacetimes is much harder and of a different nature. The main differences and the basic complications arise due to the fact that $\bm n$ is now timelike, and thus $\scri$ is a spacelike hypersurface: there is no notion of `evolution'. The topology of $\scri$ is not determined, and it has no `universal' structure. The existence of infinitesimal symmetries is not guaranteed. There is a big issue concerning in- and out-going gravitational radiation. The very notion of energy is unclear because there cannot be any globally defined timelike Killing vector ---actually all possible Killing vectors on $(\hat M,\hat g)$ become tangent to $\scri$ at $\scri$.
And there are other issues, see e.g., \cite{Penrose2011,Ashtekar2017, Ashtekar2014,Szabados2019}. Still, criteria \ref{crit1} and \ref{crit2} appropriately identify the cases without radiation, even though there remain some subtleties to be understood concerning the mixture (or possible anihilation) of in- and out-going radiation.

Let us start by noticing that, contrary to the asymptotically flat case where generally one deals with a nice topology $\mathbb{R}\times \mathbb{S}^2$, in the case with $\Lambda>0$ the topology of any connected component of $\scri$ is not determined, figure \ref{fig:scri}. Its topology can be (see e.g. \cite{Mars2017} with examples)
\begin{figure}
\includegraphics[width=12cm]{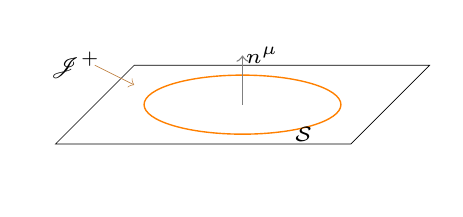}
\caption{This is a schematic representation of $\scri^+$ when $\Lambda >0$, where $n^\mu$ is timelike and normal to $\scri^+$, and $\Sc$ represents a cut with spherical topology. As usual, one dimension is suppressed. The topology of $\scri$ is not fixed, the manifold can be $\mathbb{R}^3$, $\mathbb{R}\times\mathbb{S}^2$, $\mathbb{S}^3$, or even $\mathbb{S}^3\setminus \{p_1,\dots, p_n\}$ with $n>2$, see the main text. If, for instance, the topology is $\mathbb{S}^3$, the shown schematic representation should be understood as a stereographic projection onto Euclidean space. Thus, the best way to always imagine $\scri^+$ when $\Lambda >0$ is as $\mathbb{S}^3$, possibly with a number of points removed.\label{fig:scri}}
\end{figure}
\begin{enumerate}
\item $\mathbb{S}^3$. This is the case for de Sitter or Taub-NUT-de Sitter spacetimes.
\item $\mathbb{R}\times \mathbb{S}^2$. This happens in Kerr-de Sitter spacetime, including Kottler with spherical symmetry.
\item $\mathbb{R}^3$, such as in Kottler spacetimes with non-positively curved group orbits.
\item Others, $\mathbb{S}^3\setminus \{p_1,\dots, p_n\}$ with $n>2$.
\end{enumerate}
The {\em conformal} geometry of $(\scri,h_{ab})$ is given by the completion of the physical spacetime. In particular 
\begin{itemize}
\item its intrinsic Schouten tensor, which actually coincides with the pull-back of the Schouten tensor on $(M,g)$:
$$\overline{S}_{ab}:=\overline{R}_{ab}-\frac{\overline{R}}{4} h_{ab} \eqs S_{\mu\nu}e^\mu{}_a e^\nu{}_b$$
\item and the corresponding Cotton-York tensor $C_{ab}$, which coincides with the {\em magnetic} part of the re-scaled Weyl tensor \cite{Kroon,Ashtekar2014,Friedrich1986a}
\be\label{eq1}
\ele C_{ab}:= \epsilon_a{}^{cd} \overline\nabla_c \overline{S}_{db} \eqs \stackrel{*}{d}_{\mu\nu\rho}{}^\sigma \n_\sigma e^\mu{}_a \n^\nu e^\rho{}_b
\ee
where $\n^\mu $ is the normalized version of $n^\mu$.
\end{itemize}

%Notice the (subtle but important) difference between this equation and the corresponding one when $\Lambda =0$.

Only the trace-free part of $\overline{S}_{ab}$ enters into the previous equation.
Given the foliation by spacelike hypersurfaces $\Omega =$  const.\ around $\scri$ determined by $\bm{n}=d\Omega$, the time derivative of its shear $\sigma_{\mu\nu}$ coincides, on $\scri$, with the mentioned trace-free part:
$$
\dot \sigma_{ab}:\eqs  e^\mu{}_a e^\nu{}_b \lied_{\n} \sigma_{\mu\nu}=\overline{S}_{ab}-\frac{1}{12} \overline{R} h_{ab}  .
$$
The completion of the physical spacetime also provides the {\em electric} part of the re-scaled Weyl tensor\footnote{The standard notation for this electric part is $D_{ab}$  \cite{Fernandez-Alvarez_Senovilla20b,Fernandez-Alvarez_Senovilla-dS,Friedrich1986a,Friedrich1986b,Friedrich2002}, but I will use $\F_{ab}$ herein to avoid notational conflicts.}
$$
\F_{ab} :\eqs d_{\mu\nu\rho}{}^\sigma \n_\sigma e^\mu{}_a \n^\nu e^\rho{}_b
$$
but this is not intrinsic to $(\scri,h_{ab})$. $\F_{ab}$ can be seen to coincide with the second time-derivative of the shear: 
$$
\ddot \sigma_{ab} \eqs 2\ele \F_{ab}.
$$
In general, $C_{ab}$ and $\F_{ab}$ are {\em trace-free} tensors with gauge behaviour under \eqref{gauge}
$$
\{C_{ab},\F_{ab}\} \rightarrow \omega^{-1} \{C_{ab}, \F_{ab} \}.
$$
From the Bianchi identities, $C_{ab}$ is also divergence free, that is to say, it is a TT-tensor. For appropriate decaying condition of the physical energy-momentum tensor, $\F_{ab}$ is also a TT-tensor. Under these decaying conditions the Bianchi identities reduce to
\be
\nablah_a C^{ab}=0, \hspace{4mm} \nablah_a \F^{ab}=0, \hspace{4mm}
\nablah_{[c} C_{a]b} =\frac{1}{2} \epsilon_{cad}\dot{\F}^d{}_b, \hspace{4mm}
\nablah_{[c} \F_{a]b} =\frac{1}{2} \epsilon_{cad}\dot{C}^d{}_b .\label{Bianchi}
\ee
Note that the first two are consequences of the second pair by using the traceless property of $\F_{ab}$ and $C_{ab}$. In the above, the dot means derivative along the unit normal to $\scri$.%, more especifically, $\dot\F_{ab}$ means the projection to $\scri$ of the $\vec n$-derivative of the electric part $\hat{\F}$. 

There are several fundamental results demonstrating that the geometry of the physical spacetime is fully encoded, as initial conditions of a well-posed initial value problem, on $(\scri,h_{ab})$ {\em together} with a symmetric and trace-free tensor field ($\F_{ab}$). This can be seen as an initial or final value problem. Specifically, I refer to
\begin{itemize}
\item A classical result by Starobinsky \cite{Starobinsky1983}. An expansion in powers of $e^{-\ele t}$ as $t\rightarrow \infty$ shows that the first term is a spatial 3-dimensional metric $h_{ab}$, then the next two terms are determined by the curvature of $h_{ab}$ and a traceless symmetric tensor $\F_{ab}$ whose divergence depends on the matter contents --and is divergence free in vacuum--, and these three terms determine the whole expansion.
\item A more mathematical (and more general) similar result due to Fefferman and Graham \cite{Fefferman2002,Fefferman2005} showing that given any conformal geometry $(\S,h_{ab})$ the addition of a TT-tensor $\F_{ab}$ provides, via a well determined expansion, a 4-dimensional spacetime whose conformal completion has $(\scri ,h_{ab})=(\S,h_{ab})$.
\item The results by Friedrich \cite{Friedrich1986a,Friedrich1986b,Friedrich2002,Kroon} proving that the $\Lambda$-vacuum Einstein field equations are equivalent to a set of symmetric hyperbolic partial differential equations on the unphysical spacetime and the solutions are fully determined by initial/final data consisting of a 3-dimensional Riemannian manifold with the metric conformal class plus a TT-tensor. The Riemannian manifold turns out to be (a representative of the conformal class of) $(\scri,h_{ab})$ while the TT-tensor coincides with the electric part $\F_{ab}$ of the re-scaled Weyl tensor.
\end{itemize}
 In summary, we now know that {\em any property of the physical spacetime is fully encoded in the triplet $(\scri,h_{ab},\F_{ab})$}. Consequently, the existence, or absence, of gravitational radiation {\em is also fully encoded} in $(\scri,h_{ab},\F_{ab})$. Our criteria fulfil this completely, because the
 asymptotic super-momentum can be split into the parts tangent and normal to $\scri$
$$
p^\alpha := -\D^\alpha{}_{\beta\mu\nu} n^\beta n^\mu n^\nu \eqs \overline W \n^\alpha +\bar p^a e^\alpha{}_a
$$ 
and  \eqref{divPi}, that now requieres appropriate matter decaying conditions, gives
\be\label{continuity}
\nabla_\mu p^\mu \eqs 0 \hspace{2mm} \Longrightarrow \hspace{2mm} \dot{\overline W} +\overline\nabla_a \bar p^a=0.
\ee
$\bar p^a$ is called the {\em asymptotic super-Poynting} vector.
Observe that criterion \ref{crit1} (respectively criterion \ref{crit2}) states that there is no gravitational radiation crossing a cut $\Sc\subset \scri$ (resp.\ $\Delta$) if $\bar p^a$ vanishes on $\Sc$ (resp.\ $\Delta$). From well-known old results \cite{Bel1962,Maartens1998,Alfonso2008}
\be
\bar p_a=2\left(\frac{\Lambda}{3}\right)^{(3/2)} \epsilon_{abc} C^{bd} \F^c{}_d \label{superp}
\ee
so that there is no gravitational radiation crossing $\scri$ if and only if $C^a{}_b$ and $\F^a{}_b$ conmute:
$$
 \bar p_a =0 \hspace{4mm} \Longleftrightarrow \hspace{4mm}  \epsilon_{abc} C^{bd} \F^c{}_d =0.
$$
This condition is truly encoded on $(\scri,h_{ab},\F_{ab})$ and it takes all its elements into account, as required.
\begin{Remark}[Radiation encoded at $\scri$]
From the perspective of the initial, or final, value problem, given a particular conformal geometry representing $(\scri,h_{ab})$, one only needs to add a TT tensor $\F_{ab}$ such that it does (not) conmute with the Cotton-York tensor $C_{ab}$ if the spacetime is going to (not) be free of gravitational radiation. Observe that there is a special possibility when $(\scri,h_{ab})$ is conformally flat, so that $C_{ab}=0$, in which case no matter which TT-tensor field $\F_{ab}$ one adds the resulting spacetime will not contain gravitational radiation.
\end{Remark}

Let now $\Delta\subset \scri$ be an open region of $\scri$ bounded by two disjoint cuts $\Sc_1$ and $\Sc_2$, as shown in figure \ref{fig:2cuts}. From \eqref{continuity} one easily gets 
\be\label{balance}
\int_\Delta \dot{\overline{W}}\bm{\epsilon} = \int_{\Sc_1} m^a_1 \bar p_a \bm{\epsilon}_2 -\int_{\Sc_2} m^a_2 \bar p_a \bm{\epsilon}_2
\ee
where $m^a_1$ and $m^a_2$ are the unit normals to $\Sc_1$ and $\Sc_2$ within $\scri$, respectively. We will later see that $\bar{p}_a m^a$ has a sign in relevant cases.
\begin{figure}[!ht]
\includegraphics[width=12cm]{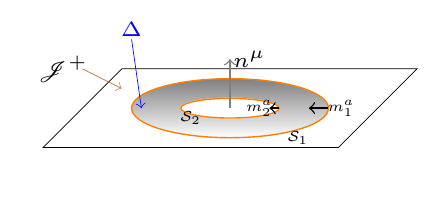}
\caption{Schematic representation of a region $\Delta$ in $\scri^+$ when $\Lambda >0$ bounded by two disjoint cuts $\Sc_1$ and $\Sc_2$. The vector fields $m_1^a$ and $m_2^a$ are the unit normal vectors to the cuts $\Sc_1$ and $\Sc_2$ within $\scri^+$, respectively. \label{fig:2cuts}}
\end{figure}

\subsection{Geometry of cuts on $\scri$}
Our criteria for absence of radiation are primarily associated to cuts, and thus it is convenient to develop some formalism for the geometry of these cross-sections of $\scri$ in relation with the physical quantities relevant for the criteria. Let $\Sc$ be any cut on $\scri$ and let $m^b$ denote the unit vector field normal to ${\cal S}$ within $\scri$ and, as before, $\{E_A^a\}$ a basis of tangent vector fields on ${\cal S}$. The first fundamental form of the cut is denoted by
$$
q_{AB}= h_{ab} E^a_A E^b_B
$$
and \eqref{gaugeh} still holds now. Define for every symmetric tensor field $\bar{t}^{ab}$ on $\scri$ its corresponding parts in an orthogonal decomposition relative to $\Sc$ and thereby introduce the notation for all such tensor decompositions:
$$
\bar{t}^{ab} = t^{AB} E^a_A E^b_B +t^A E^a_A m^b + t^B E^b_B m^a + t m^a m^b
$$ 
and then raise and lower indices of the objects on ${\cal S}$ with the inherited metric $q_{AB}$. The Levi-Civita connection of $({\cal S},q_{AB})$ is denoted by $\gamma^A_{BC}$ and one then has
$$
E^a_A \nablah_a E^b_B = \gamma^C_{AB} E^b_C -\varkappa_{AB} m^b, 
$$
where $\varkappa_{AB}$ is the 2nd fundamental form of ${\cal S}$ in $\scri$ ---and also the unique non-zero 2nd fundamental form of ${\cal S}$ in the unphysical spacetime. One can decompose this object as usual
$$
\varkappa_{AB}:=\Sigma_{AB} +\frac{1}{2} \varkappa q_{AB}, \hspace{7mm} \varkappa := q^{AB}\varkappa_{AB} , \hspace{7mm} q^{AB}\Sigma_{AB}=0 
$$
where $\Sigma_{AB}$ is the shear of ${\cal S}$ in $\scri$ ---or the unique non-zero shear of ${\cal S}$ in the unphysical spacetime. 
Furthermore, for any symmetric $t_{ab}$
$$
E^a_A E^b_B E^c_C  \nablah_c \bar{t}_{ab} = D_C t_{AB} +t_A \varkappa_{BC} +t_B \varkappa_{AC} .
$$

Under the allowed gauge transformations \eqref{gaugeh} the above objects and those relative to $\overline{S}_{ab}$ transform as follows ($\omega_A := D_A\omega$, $\omega_m := m^b \nablah_b \omega$)):
\bea 
\tilde{m}_a &=&\omega m_a,\\
\tilde{\gamma}^C_{AB} &=&\gamma^C_{AB} +\frac{1}{\omega}\left(\delta^C_A\omega_B +\delta^C_B \omega_A-\omega^C q_{AB} \right),\label{gammagauge}\\
\tilde{\varkappa}_{AB}&=& \omega \varkappa_{AB} + \omega_m q_{AB},\label{kappagauge}\\
\tilde{\Sigma}_{AB} &=& \omega \Sigma_{AB},\label{Sigmagauge}\\
\tilde{\varkappa} &=&\frac{1}{\omega} \varkappa + \frac{2}{\omega^2} \omega_m,\label{trkappagauge}\\
\tilde{S}_{AB}&=& S_{AB}-\frac{1}{\omega}D_A\omega_B +\frac{2}{\omega^2} \omega_A\omega_B -\frac{1}{2\omega^2} \omega^D\omega_D q_{AB}-\frac{\omega_m}{\omega} \left(\varkappa_{AB}+\frac{1}{2\omega} \omega_m q_{AB}\right),\label{gaugeS'}\\
\tilde{S}_A &=& \frac{1}{\omega} \left(S_A -\frac{1}{\omega} D_A\omega_m+\frac{1}{\omega}\varkappa_{AB}\omega^B +\frac{2}{\omega^2} \omega_m \omega_A \right),\\
\tilde{S} &=& \frac{1}{\omega^2} \left(S- \frac{1}{\omega} m^a m^b \nablah_a \nablah_b \omega +\frac{2}{\omega^2} \omega_m^2 -\frac{1}{2\omega^2} \nablah_c\omega \nablah^c\omega\right).
\eea

The projections of the gauge-invariant equation (\ref{eq1}) onto to cut ${\cal S}$ lead to the following relations
\bea
D_{[C}S_{A]B}+\varkappa_{B[C} S_{A]} =\frac{1}{2}\ele \epsilon_{CA} C_B, \label{eq11}\\
E^a_A E^b_B m^c  \nablah_c S_{ab} -D_A S_B+\varkappa_A^D S_{BD}-S \varkappa_{AB}= \ele\epsilon_A{}^D C_{DB}  \label{eq12}
\eea
where $\epsilon_{AB}$ is the canonical volume element 2-form on $({\cal S},q_{AB})$. Relation (\ref{eq11}) is gauge invariant, while (\ref{eq12}) is gauge homogeneous with a factor $1/\omega$. As the righthand side of (\ref{eq11}) is easily seen to be gauge invariant (because $\tilde{C}_{ab}=(1/\omega) C_{ab}$), it follows that $D_{[C}S_{A]B}+\varkappa_{B[C} S_{A]} $ is also gauge invariant. The skew-symmetric part of (\ref{eq12}) reads
$$
D_{[C} S_{A]} -\varkappa^D_{[C}S_{A]D}=\frac{1}{2} \ele \epsilon_{CA} C
$$
(notice that $C:=C_{ab}m^am^b =-C^E_E$, as follows from $C^b_b=0$),  while the symmetric part reads
$$
E^a_A E^b_B m^c  \nablah_c S_{ab} -D_{(A} S_{B)}+\varkappa_{(A}^D S_{B)D}-S \varkappa_{AB}= \ele\epsilon_{(A}{}^D C_{B)D}=\ele\epsilon_{A}{}^D \hat{C}_{BD}
$$
where we use a hat over the matrices to denote its trace-free part:
\be
\hat{C}_{AB} := C_{AB} -\frac{1}{2} q_{AB} C^E{}_E, \hspace{1cm} \epsilon_{DA}\hat{C}_{B}{}^D=\epsilon_{DB}\hat{C}_{A}{}^D = \epsilon_{D(A} C_{B)}{}^D 
\ee
and similarly for $\hat{\F}_{AB}$. Using the 2-dimensional identity
$$
\varkappa_{(A}^DS_{B)D} -\frac{1}{2} \varkappa S_{AB} -\frac{1}{2} S^D_D \varkappa_{AB}+\frac{1}{2}\left(\varkappa S^D_D-\varkappa^{CD} S_{CD} \right)q_{AB}=0
$$
the previous symmetric part can be recast into the form
\bea
E^a_A E^b_B m^c  \nablah_c S_{ab} -D_{(A} S_{B)}+\frac{1}{2} \varkappa S_{AB}+\left(\frac{1}{2} S^D_D-S\right) \varkappa_{AB}\nonumber \\
-\frac{1}{2}\left(\varkappa S^D_D-\varkappa^{CD} S_{CD} \right)q_{AB}= \ele\epsilon_{A}{}^D \hat{C}_{BD} .
\eea
 An equivalent form of (\ref{eq11}) is
$$
D_B S^B_A -D_A S^D_D +\varkappa S_A -\varkappa_A^B S_B =\ele C^D \epsilon_{DA}.\label{eq11'}
$$

 %Thus, one can try to apply the usual reasoning to $S_{AB}$, whose gauge behaviour (\ref{Sgauge}) is complicated, and check whether there can be a gauge invariant part contained in $S_{AB}$.

One can rewrite (\ref{eq11}) in a form without $S_A$. This can be achieved by using the Gauss and Codazzi relations for ${\cal S}$, which can be checked to read
\bea
S_{A[C}q_{D]B}+q_{A[C}S_{D]B}&=&Kq_{A[C}q_{D]B}-\varkappa_{A[C}\varkappa_{D]B}, \label{gauss0}\\
D_{[C}\varkappa_{A]B}&=&q_{B[C}S_{A]} \label{cod}
\eea
Relation (\ref{cod}) is equivalent to its trace
\be
S_A=D_E\varkappa^E_A-D_A\varkappa . \label{cod'}
\ee
The Gauss equation (\ref{gauss0}) is also fully equivalent to its trace and also to its double trace
\bea
S^D_D q_{AB} &=&K q_{AB} +\varkappa_A^D\varkappa_{DB} -\varkappa \varkappa_{AB},\label{gauss}\\
S^D_D &=&K +\frac{1}{2}\left(\varkappa^{AB}\varkappa_{AB}- \varkappa^2 \right)=K-\det(\varkappa^E_F) \label{gauss'},
\eea
which can be easily checked by using a typical 2-dimensional identity, and for the last part also using the Caley-Hamilton theorem 
$$
\varkappa_A^D\varkappa_{DB} -\varkappa \varkappa_{AB}+q_{AB} \det(\varkappa^E_F)=0.
$$
Another simpler version of this relation is simply
\be
\Sigma_A{}^D\Sigma_{DB} =\frac{1}{2} \Sigma_{DE} \Sigma^{DE} q_{AB} . \label{Sigmasquare}
\ee
Notice that
$$
\varkappa_{AB} \varkappa^{AB} = \Sigma_{AB} \Sigma^{AB} +\frac{1}{2} \varkappa^2 .
$$
Using \eqref{cod'}, equation (\ref{eq11}) can be rewritten as
\be
D_{[C}S_{A]B}+\varkappa_{B[C}\left(D^E\varkappa_{A]E} -D_{A]}\varkappa \right)=\frac{1}{2}\ele \epsilon_{CA} C_B, %\label{eq11}
\ee
whose lefthand side is (must be!) gauge invariant, in accordance with (\ref{DSgauge}). This is still equivalent, after some calculation, to
\bea
D_C \left(S^C{}_A -\frac{1}{2} \Sigma^{CE}\Sigma_{EA}+ \frac{\varkappa}{2}\Sigma^C{}_A +\frac{\varkappa^2}{8} \delta^C_A -K\delta^C_A \right)=\nonumber \\
\frac{3}{2}D_B(\Sigma^{BE}\Sigma_{EA})-\Sigma^{CE}D_E \Sigma_{CA} +\ele \epsilon_{EA}C^E\label{eq11''}.
\eea
Observe that the righthand side in this expression is gauge homogeneous with a factor $1/\omega^2$.

Projecting the Bianchi equations (\ref{Bianchi}) to the cut ${\cal S}$ as before one derives
\bea
D_{[C}C_{A]B} +\varkappa_{B[C} C_{A]} =\frac{1}{2} \epsilon_{CA} \dot\F_B,\\
E^a_A E^b_B m^c\nablah_c C_{ab} -D_A C_B +\varkappa_A{}^D C_{BD} +C^E{}_E \varkappa_{AB} =\epsilon_{AD} \dot\F^D{}_B,\\
D_{[C}\F_{A]B} +\varkappa_{B[C} \F_{A]} =\frac{1}{2} \epsilon_{CA} \dot{C}_B,\\
E^a_A E^b_B m^c\nablah_c \F_{ab} -D_A \F_B +\varkappa_A{}^D \F_{BD} +\F^E{}_E \varkappa_{AB} =\epsilon_{AD} \dot{C}^D{}_B
\eea

Analogously to Lemma \ref{lem:dt2} one can prove the following result for cuts on $\scri$ when $\Lambda >0$
\begin{Lemma}\label{lem:dt1}
Let $p_{AB}=p_{(AB)}$ be any symmetric tensor field on $({\cal S},q_{AB})$ whose gauge behaviour under residual gauge transformations (\ref{gaugeh}) is
$$
\tilde{p}_{AB}= p_{AB}-\frac{1}{\omega}D_A\omega_B +\frac{2}{\omega^2} \omega_A\omega_B -\frac{1}{2\omega^2} \omega^D\omega_D q_{AB}-\frac{\omega_m}{\omega} \left(\varkappa_{AB}+\frac{1}{2\omega} \omega_m q_{AB}\right)
$$
Then, 
\bean
\tilde{D}_{[C}\tilde{p}_{A]B}+\tilde{\varkappa}_{B[C}\left(\tilde{D}^E\tilde{\varkappa}_{A]E}-\tilde{D}_{A]}\tilde{\varkappa} \right)&=&D_{[C}p_{A]B}+\varkappa_{B[C}\left(D^E\varkappa_{A]E}-D_{A]}\varkappa \right) \\
&+&\frac{1}{\omega} \left(p_{B[C}-S_{B[C}\right)\omega_{A]}+\frac{1}{\omega}q_{B[C} \left(p^D_{A]}-S^D_{A]} \right)\omega_D 
\eean
\end{Lemma}
The proof is again by direct calculation. As a corollary we immediately have
\be
\tilde{D}_{[C}\tilde{S}_{A]B}+\tilde{\varkappa}_{B[C}\left(\tilde{D}^E\tilde{\varkappa}_{A]E}-\tilde{D}_{A]}\tilde{\varkappa} \right)=D_{[C}S_{A]B}+\varkappa_{B[C}\left(D^E\varkappa_{A]E}-D_{A]}\varkappa \right) \label{DSgauge}
\ee

\subsubsection{The super-Poynting vector and asymptotic radiant super-momenta on cuts of $\scri$}\label{subsec:superp}
Let me denote by
$$
\vec k_\pm := \vec \n \pm \vec m, \hspace{1cm} k^\mu_+ k_{-\mu}=-2
$$
the two future null normals to the cut $\Sc$ (see figure \ref{fig:kpm}) and, given that $\S_{AB}$ is the only non-zero shear of $\Sc$ in $\scri$, the corresponding two null shears are simply $\pm\S_{AB}$.
\begin{figure}[!ht]
\includegraphics[width=12cm]{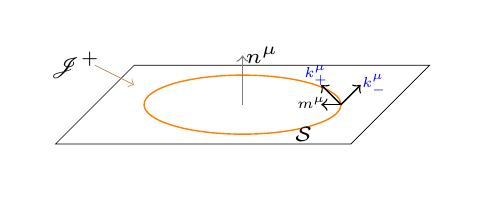}
\caption{Schematic representation of the two null normals $k^\mu_\pm =\n^\mu \pm m^\mu$ to the cut $\Sc$ at a given point of the cut. \label{fig:kpm}}
\end{figure}
We introduce, for each cut $\Sc$, the two {\em asymptotic radiant super-momenta} as
\be\label{Qpm}
Q_\pm^\alpha := -D^\alpha{}_{\mu\nu\rho} k_\pm^\mu k_\pm^\nu k_\pm^\rho, 
\ee
and they are always, by construction, null and future. It is convenient to have formulae for $\bar p_a$ and also for $Q_\pm^\alpha$ in terms of $C_{ab}$ and $\F_{ab}$. To that end, we write the asymptotic  radiant super-momenta in the given bases
\be\label{Qs}
Q_\pm^\alpha = \frac{1}{2} \W_\pm k_\mp^\alpha +\frac{1}{2} \Z_\pm k^\alpha_\pm +Q^A_\pm E^\alpha_A
\ee
or equivalently
\be\label{Qs1}
Q_\pm^\alpha = \frac{1}{2}(\W_\pm +\Z_\pm) \n^\alpha \pm \frac{1}{2} (\Z_\pm -\W_\pm) m^\alpha 
+Q^A_\pm E^\alpha_A
\ee
where by direct (long) calculation one finds
\bea
\W_\pm &:=& -k_\alpha^\pm Q_\pm^\alpha =8(\hat{\F}_{AB}\mp \epsilon_{DA}\hat{C}_{B}{}^D)
(\hat{\F}^{AB}\mp \epsilon^{CA}\hat{C}^{B}{}_C)\geq 0, \label{Ws}\\
\Z_\pm &:= & -k_\alpha^\mp Q_\pm^\alpha = 4(\F_A\pm\epsilon_{AB}C^B)(\F^A\pm\epsilon^{AD}C_D)\geq 0,
\label{Zs}\\
Q_\pm^A &:=& W^A_\alpha Q_\pm^\alpha =\pm 8 (\hat{\F}_{AB}\mp \epsilon_{D(A}\hat{C}_{B)}{}^D)(\F^B\pm\epsilon^{BE}C_E)\label{QAs}.
\eea
Some useful formulas are
\bea
\Z_+ -\Z_- &=& 16\epsilon_{AB}\F^AC^B,\hspace{1cm} \Z_+ +\Z_- = 8(\F_A\F^A+C_A C^A),\label{Zs1}\\
\W_+ -\W_- &=& %32\epsilon_{AB}\F^{AD}C^B{}_D=
32\epsilon_{AB}\hat{\F}^{AD}\hat{C}^B{}_D, \hspace{3mm} \W_+ +\W_- =16(\hat{\F}_{AB}\hat{\F}^{AB}+\hat{C}_{AB}\hat{C}^{AB}),\label{Ws1}\\
Q_+^A -Q_-^A &=& 16 (\hat{C}^{AB} C_B + \hat{\F}^{AB}\F_B), \, \, 
Q_+^A +Q_-^A = 16\epsilon_{AB} (\hat{C}^{BD} \F_D - \hat{\F}^{BD}C_D).
\eea
Then, the expressions of the components of $\bar p_a$ can be easily found. Orthogonally decomposing the super-Poynting on $\Sc$ as
$$
\left(\frac{3}{\Lambda}\right)^{3/2}\bar p^a = p_m m^a + p^A E_A^a 
$$
another straightforward calculation leads to
\be\label{pm}
p_m =\frac{1}{16} \left(\Z_+-\Z_--\W_++\W_-\right)+3\epsilon_{AB}C^A\F^B=\frac{1}{16} \left(\W_--\W_++2\Z_--2\Z_+ \right)
\ee
(where the first in (\ref{Zs1}) has been used) and to
\bea
p_A &=& 2\epsilon_{AB} \left(C^{BD}\F_D -\F^{BD} C_D +C^E{}_E \F^B-\F^E{}_E C^B \right)\nonumber\\
&=& 2\epsilon_{AB} \left(\hat{C}^{BD}\F_D -\hat{\F}^{BD} C_D +\frac{3}{2}C^E{}_E \F^B-\frac{3}{2}\F^E{}_E C^B \right) \label{pA}\\
&=& \frac{1}{8} \left(Q^+_A +Q^-_A\right)+3\epsilon_{AB}\left(C^E{}_E \F^B-\F^E{}_E C^B\right).\nonumber
\eea
For completeness, we note in passing that
\be\label{sumQs}
Q_+^\alpha+Q_-^\alpha= \frac{1}{2} (\W_++\W_-+\Z_+ +\Z_-) n^\alpha +\frac{1}{2}(\Z_+ -\Z_- -\W_++\W_-)m^\alpha+
\left(Q_+^A+Q_-^A \right)E^\alpha_A .
\ee

\section{Are there any News for cuts (and for $\scri$)?}\label{sec:news}
There are some objects in the literature that are called ``news'' tensor in the case with $\Lambda >0$ based on analogies with the asymptotically flat case. None of them seem to have led to properties similar to that of the News tensor when $\Lambda=0$, and one can raise some doubts about the existence of news in the general case with $\Lambda >0$. Nevertheless, in this section I describe a general method to search for such `News', and also a tensor field is uncovered that will certainly be part of any news tensor, if this exists. 

Recall first of all that, when $\Lambda =0$, $N_{ab}$ is the pull-backed Schouten tensor gauge corrected, and that one can unambiguously define the news tensor associated to any cut $\Sc$ by projecting into the cut. An interesting idea, given the previous considerations, is to try to assign to any possible cut ${\cal S}\subset \scri$ ---and especially when the cut is topologically $\mathbb{S}^2$--- a gauge invariant tensor field contained {\em partly} in the pullback to ${\cal S}$ of $\overline{S}_{ab}$. 

Why partly? Well, there are crucial differences now with respect to the case with $\Lambda =0$, as now the Schouten tensor $\overline{S}_{ab}$ is fully intrinsic to $(\scri,h_{ab})$, in contrast with the asymptotically flat case where it arises as the curvature of the connection, and this is inherited from the ambient manifold but not intrinsic to the null $(\scri,h_{ab})$. In this sense, note that \eqref{eq1} is fully intrinsic to the spacelike $(\scri,h_{ab})$ showing in particular that $\overline{S}_{ab}$ is {\em determined exclusively by $C_{ab}$} and thus {\em it cannot contain by itself any gauge-invariant part that describes the existence of radiation, which as explained before, must be encoded in the triplet $(\scri,h_{ab},\F_{ab})$.}  A key equation now is the identity
$$
\frac{1}{2} \frac{3}{\Lambda} \bar p_c = \overline\nabla_c (\F^{ab} \overline{S}_{ab} ) -\overline\nabla_a (\F^{ab}\overline{S}_{bc}) -\overline{S}_{ab} \overline\nabla_c \F^{ab}
$$
which graphically shows that the asymptotic super-Poynting depends on the interplay between $\overline{S}_{ab}$ and $\F_{ab}$. In this formula, every term on the righthand side has a complicated gauge behaviour yet their combination equals $\bar p_c$, whose gauge behaviour is simply $\bar p_c \rightarrow \omega^{-5} \bar p_c$. Given that the vanishing of $\bar p_c$ characterizes the absence of radiation, the existence of any `source' of type News for $\bar p_c$ requires a splitting of the righthand terms in gauge well-behaved parts plus a remainder that must be uniquely determined. Such a ``News tensor'' should then satisfy appropriate differential equations.

Despite these difficulties, $\overline{S}_{ab}$ will probably entail the part of the news (if this exists) not related to the TT-tensor $\F_{ab}$. This is the part that we were able to identify \cite{Fernandez-Alvarez_Senovilla-dS}, as I discuss in the following.

Let us generalize Corollary \ref{coroRho} by finding the general form of the tensor fields defined by Corollary \ref{coroDt} but with a general, non-vanishing, $D_{[C}t_{A]B}$.
\begin{Proposition}
Let $\Sc\subset \scri$ by a cut on $\scri$. If the equation 
\be\label{DU}
D_{[C}W_{A]B}= X_{CAB} 
\ee
for a given \underline{{\em gauge invariant}} tensor field $X_{CAB}=X_{[CA]B}$ has a solution for $W_{AB}=W_{(AB)}$ whose gauge behaviour is (\ref{normalgauge}) with $a=1$, then this solution is given by
\be\label{U}
W_{AB} = S_{AB} -\frac{1}{2} \Sigma_{A}{}^D \Sigma_{BD} +\frac{\varkappa}{2}\Sigma_{AB} +\frac{\varkappa^2}{8} q_{AB}+M_{AB}
\ee
where $M_{AB}$ is a trace-free, gauge invariant and symmetric tensor field solution of
\be
D_{[C}M_{A]B}= X_{CAB} -\frac{1}{2}\ele \epsilon_{CA} C_B +D_{[C} \left(\Sigma_{A]E}\Sigma_{B}{}^E \right)-\frac{1}{2} D_B \Sigma_{[C}{}^E \Sigma_{A]E}.\label{DM}
\ee
\end{Proposition}
{\bf Remark}: The righthand side of (\ref{DM}) is gauge invariant. If the cut has $\mathbb{S}^2$ topology the solution is unique. More generally, $M_{AB}$ (and a fortiori $W_{AB}$) is unique whenever $(\Sc,q_{AB})$ has a conformal Killing vector with a fixed point \cite{Fernandez-Alvarez_Senovilla-dS}.
\begin{proof}
By using (\ref{gammagauge}), (\ref{Sigmagauge}), (\ref{trkappagauge}) and (\ref{gaugeS'}) it is a matter of checking that the tensor \eqref{U} has the gauge behaviour (\ref{normalgauge}) with $a=1$, provided $M_{AB}$ is gauge invariant. Its trace, on using (\ref{gauss'}) and (\ref{Sigmasquare}) is 
\be\label{trW}
W^E{}_E =K.
\ee
 Therefore, Corollary \ref{coroDt} applies and $D_{[C}W_{A]B}$ is gauge invariant.
For the second part, using (\ref{eq11''}) and manipulating a little one arrives at
$$
D_{[C}W_{A]B}= \frac{1}{2}\ele \epsilon_{CA} C_B -D_{[C} \left(\Sigma_{A]E}\Sigma_{B}{}^E \right)+\frac{1}{2} D_B \Sigma_{[C}{}^E \Sigma_{A]E}+D_{[C}M_{A]B}
$$
from where (\ref{DM}) immediately follows. Due to the second part in Corollary \ref{coroDt} $D_{[C}M_{A]B}$ is gauge invariant.
\end{proof}

Now, notice that the tensor field $W_{AB}-M_{AB}$, that is,
$$
U_{AB}:=S_{AB} -\frac{1}{2} \Sigma_{A}{}^D \Sigma_{BD} +\frac{\varkappa}{2}\Sigma_{AB} +\frac{\varkappa^2}{8} q_{AB}
$$
has the following trace
\be\label{trU}
U^E_E = K
\ee
and that equation (\ref{eq11''}) can be rewritten, in terms of $U_{AB}$ as
\be\label{paso}
D_C(U^C{}_A -U^E{}_E \delta^C_A)=\frac{3}{2}D_B(\Sigma^{BE}\Sigma_{EA})-\Sigma^{CE}D_E \Sigma_{CA} +\ele \epsilon_{EA}C^E .
\ee
Contracting this equation with any conformal Killing vector field $\xi^A$ and integrating its lefthand side on $\Sc$ 
\bean
\int_\Sc \xi^A [D_C(U^C{}_A -U^E{}_E \delta^C_A)] = \int_\Sc D_C [\xi^A(U^C{}_A -U^E{}_E \delta^C_A)] - \int_\Sc (U^{CA} -U^E{}_E q^{CA})D_C\xi_A\\
=  \int_\Sc D_C [\xi^A(U^C{}_A -U^E{}_E \delta^C_A)] - \frac{1}{2} \int_\Sc (U^{CA} -U^E{}_E q^{CA})q_{CA} D_B\xi^B\\
= \int_\Sc D_C [\xi^A(U^C{}_A -U^E{}_E \delta^C_A)] +\frac{1}{2} \int_\Sc K D_B\xi^B
\eean
where in the last equality I have used \eqref{trU}. If $\Sc$ is compact the first summand here vanishes. Concerning the second, a non-trivial result proved in Appendix \ref{App:rho}, namely (\ref{intLieK}), shows that this term also vanishes if $\Sc$ is compact. Therefore, whenever the cut $\Sc$ is compact we arrive at
\be\label{CKV}
\int_{\cal S} \xi^A\left(\frac{3}{2}D_B(\Sigma^{BE}\Sigma_{EA})-\Sigma^{CE}D_E \Sigma_{CA} +\ele \epsilon_{EA}C^E \right) =0
\ee
for every conformal Killing vector fields $\xi^A$ if $\Sc$ is compact.

Define the {\em first piece of news} on $\Sc$ as the tensor field
\be\label{V}
V_{AB} := U_{AB} -\rho_{AB}
\ee
where $\rho_{AB}$ is the tensor field of Corollary \ref{coroRho}. Explicitly, the first piece of news is given by
$$
V_{AB}=S_{AB} -\frac{1}{2} \Sigma_{A}{}^D \Sigma_{BD} +\frac{\varkappa}{2}\Sigma_{AB} +\frac{\varkappa^2}{8} q_{AB}-\rho_{AB} .
$$
By construction, $V_{AB}$ is gauge invariant and trace free, so that 
$$D_{[C}V_{A]B}=D_{[C}U_{A]B}$$
 is also  gauge invariant. However, $V_{AB}$
depends {\em only} on the intrinsic geometry of $(\scri,h_{ab})$ and the cut, and therefore it simply cannot contain the desired News tensor, which must involve, as explained, ${\cal F}_{ab}$. It follows that the part described by $M_{AB}$ must be related to ${\cal F}_{ab}$, thereby bringing the information encoded in $\F_{ab}$  into the total tensor (\ref{U}). Hence, it follows that the `source' $X_{CAB}$ in the equation (\ref{DU}) has to also entail somehow ${\cal F}_{ab}$. The definition of $V_{ab}$ induces 
\be
W_{AB}= U_{AB}+M_{AB} = \rho_{AB} + V_{AB} +M_{AB} , \label{U=r+N} 
\ee
so that $M_{AB}$ is the {\em second piece of news} and the total News tensor field of the cut $\Sc$ is
\be\label{N=V+M}
N_{AB} =V_{AB} +M_{AB} .
\ee
$N_{AB}$ is symmetric, traceless, gauge invariant and satisfies the gauge invariant equation
\be\label{DN}
D_{[C}N_{A]B}= X_{CAB} .
\ee
Notice that $N_{AB}$ is partly known, as the first piece $V_{AB}$ is explicitly known for any cut $\Sc$. To find the complete news tensor one needs to identify the appropriate tensor field $X_{CAB}=X_{[CA]B}$ the provides, via \eqref{DM}, the second piece $M_{AB}$.
Thus, the problem of the existence of $N_{AB}$ reduces to the existence of a tensor field $X_{CAB}$, or equivalently of the one-form $X_A:=X^C{}_{AC}$ with
$$
X_{CAB}= 2 q_{B[C} X_{A]} , 
$$
such that the equation (\ref{DM}) has a solution for $M_{AB}$ and the vanishing of $X_A$ be equivalent, on the entire cut ${\cal S}$, to the vanishing of $N_{AB}$.

To ascertain under which circumstances such choices allow for the existence of the tensor $M_{AB}$, let us consider the trace of (\ref{DM}) which is actually equivalent to (\ref{DM}) itself:
\be\label{divM}
\frac{1}{2} D_C M^C{}_A =X_A+\frac{1}{2}\ele \epsilon_{AB}C^B-\frac{3}{8} D_A(\Sigma_{DE}\Sigma^{DE}) +\frac{1}{2} \Sigma^{CE}D_C\Sigma_{EA}.
\ee
We know that this provides the tensor field $M_{AB}$ if and only if the righthand side is $L^2$-orthogonal to every conformal Killing vector field on ${\cal S}$ (there is a 6-parameter family of these in the sphere, Appendix \ref{App:rho}). Therefore, by using here the relations (\ref{CKV}) for every conformal Killing $\xi^A$, the existence of $N_{AB}$ requires that
\be
\int_{\cal S} \xi^A X_A   =0 \label{cond}
\ee
for every conformal Killing vector $\xi^A$. An analysis of this condition is performed in Appendix \ref{App:Hodge}. Observe that, given that $X_{CAB}$ is gauge invariant, the gauge behaviour of $X_A$ is simply
\be
\tilde{X}_A = \omega^{-2} X_A \label{Xgauge} 
\ee
and therefore the statement (\ref{cond}) is gauge independent (because $\xi^A X_A \epsilon_{BC}$ is gauge invariant). 
Using here Lemma \ref{useful}, a plausible solution for $X_A$ is any one-form of the form
\be\label{possibility}
X_A = \Delta f D_A f
\ee
for a choosable function $f$ on ${\cal S}$. Observe that, due to
$$
\tilde{\Delta} f =\frac{1}{\omega^2} \Delta f, \hspace{1cm} \forall f\in C^2({\cal S})
$$
any such one-form has the correct gauge behaviour (\ref{Xgauge}) for $f$ gauge invariant. Moreover, the physical units of $X_A$ are $L^{-2}$, and thus $f$ carries no physical units. Notice finally that $X_A=0$ if and only if $f$ is constant in the sphere topology.

In principle, if one wishes that $X_A$ be related to the existence or not of radiation, so that the vanishing of a would-be news tensor field $N_{AB}$ implies the vanishing of $X_A$ and, hopefully, viceversa, the function $f$ in (\ref{possibility}) should be related to the triplet $(\scri,h_{ab},\F_{ab})$ including explicitly $\F_{ab}$. One possibility is that $f$ be a (known) function of the potentials $H_C,h_C$ and $H_\F, h_\F$ that $\hat{C}_{AB} $ and $\hat{\F}_{AB}$ possess according to formula (\ref{Hodge2}). Observe that these potentials have the right physical dimensions (a-dimensional), they do not have a simple gauge behaviour though.

\subsection{The problem of incoming and outgoing radiation: The case with $Q_-^\alpha =0$}\label{Q-=0}
As mentioned at the beginning of section \ref{sec:dS}, one of the big differences of the $\Lambda>0$-case with respecto to the $\Lambda=0$-case is the existence of possible in-coming radiation that arrives at $\scri^+$ mingling with the outgoing flux of radiation. This is a complicated matter, and there is no easy way to try to identify in- or out-going components of the radiation. It should be remarked that our criteria \ref{crit1} and \ref{crit2}, based on the vanishing of the asymptotic super-Poynting $\bar{p}^a$ in the case with $\Lambda >0$, does not discriminate between those types of radiations. The absence/presence of radiation on a cut may in general be due to a balance between several possible components, and this varies from one cut to another. This was somehow recognized time ago as a dependence of the radiative part of the field on the direction of approach to $\scri$ if $\scri$ is not a null hypersurface \cite{Penrose65,Krtous2004,Fernandez-Alvarez_Senovilla2022}.This issue is of special importance when considering isolated sources of the radiation, or sources that are confined to a compact region of the spacetime, emitting gravitational radiation. 

In the asymptotically flat scenario the lightlike character of $\scri^+$ implies that any radiation escaping from the space-time through infinity necessarily travels along lightlike directions {\em transversal} to $\scri^+$. The generators of $\scri^+$ are the only exceptions and they provide an evolution direction which can be seen as `incoming direction' and thus, radiation from the physical spacetime is exclusively outgoing. In contrast, when $\Lambda>0$ every radiation component, without exception, crosses $\scri^+$ and escapes from the space-time. In this case one needs to find physically reasonable conditions ruling out undesired radiative components, just leaving the radiation emitted by the isolated system of sources. In \cite{Ashtekar2019} a proposal to solve this problem was presented, but this relies on information from the physical spacetime. In our opinion, and according to the entire philosophy of this paper, everything happening at the portion of the physical spacetime given by the past domain of dependence of $\scri^+$ is determined by the information encoded in the triplet $(\scri^+,h_{ab},\F_{ab})$ ---plus the conformal re-scalings--- so that any `incoming radiation' or any undesired radiation components are {\em encoded in that triplet too}. I wish to stress that this is independent of the existence of multiple isolated sources emitting the radiation, or of the possibility of scattering of the radiation by other components or matter, etcetera, because {\em everything} that happens in the (domain of dependence of $\scri$ in the) physical spacetime is encoded in the initial/final data $(\scri,h_{ab},\F_{ab})$. 

\begin{figure}[!ht]
\includegraphics[width=8cm]{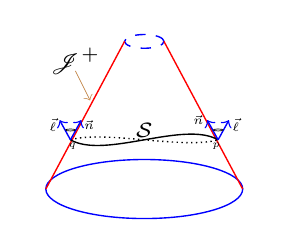}
\hspace{-1.8cm}
\includegraphics[width=10cm]{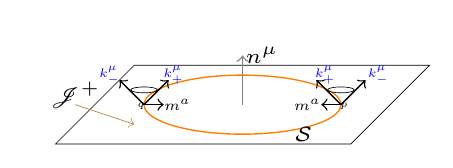}
\caption{Comparison of $\scri^+$ and null directions orthogonal to a cut $\Sc$ for the case with $\Lambda =0$ (left) and the case with $\Lambda >0$ (right). On the left the physical spacetime is the region below the cone representing $\scri^+$ and on the right the region below the plane that represents $\scri^+$. In both cases two points $p$ and $q$ belonging to the cut are shown, as well as the two null normals to the cut $\Sc$ at those points. On the left, they are given by $n^\mu$ itself, and $\ell^\mu$; on the right by $k^\mu_\pm = \n^\mu\pm m^\mu$, where $\vec m$ is the unit normal to $\Sc$ within $\scri^+$, so that $m^\mu =m^ae^\mu_a$. We know that, on the left, the vanishing of the asymptotic radiant super-momentum $\Q^\mu =0$ is equivalent to the vanishing of the news tensor and thus to the absence of radiation crossing $\scri^+$. If one modifies the cut passing through, say, $p$ the picture would be similar, but with a different $\vec\ell$. {\em All possible} such null $\vec \ell$, for all possible cuts through $p$, span the little cone shown above $p$, and similarly for $q$. Hence, vanishing of $\Q^\mu$ implies that there is no radiation on any of all those transversal directions spanning the little cone with the exception, of course, of $\vec n$, which is not transversal but tangent to $\scri^+$ and actually defines an evolution direction to the future. Notice that $\Q^\mu=0$ states that $n^\mu$ is a multiple principal null direction of the re-scaled Weyl tensor $d_{\alpha\beta\lambda}{}^\mu$. Inspiration from these properties on the left is used on the right picture to try to isolate a unique component of radiation arriving at the cut $\Sc$: when $\Lambda >0$ (right picture) set $\Q^\mu_-=0$, and assume that this implies absence of radiation arriving along the directions spanned by the little cones shown above $p$ or $q$ except along $k^\mu_-$, in analogy with the left situation. This would mean that the radiation is arriving basically along the null direction $k^\mu_-$, which again is a multiple null direction of the re-scaled Weyl tensor, so that this makes sense. If this interpretation is accepted, the vector $m^a$ on the right defines, in analogy with $n^a$ on the left, an evolution direction towards the ``future'' within the spacelike $\scri^+$. In a way, one can think that the radiation is crossing $\Sc$ towards its exterior (the projection of $k^\mu_-$). \label{fig:doble}}
\end{figure}

Moreover, one can try to get some inspiration from the asymptitcally flat situation. The vanishing of the radiant super-momentum when $\Lambda=0$ entails the absence of radiation transversal to $\scri^+$, and thus we may suspect that absence of radiation propagating {\em transversally} to some null direction is also encoded in the analogous radiant super-momenta. More specifically, in our setup the vanishing of one of the radiant supermomenta \eqref{Qpm} may mean absence of radiation components travelling along the corresponding transversal directions on that particular cut $\Sc$. This is graphically explained in figure \ref{fig:doble}.

Consider for instance the case with $\Q^\mu_- =0$ on a cut $\Sc$. By the previous discussion, this may indicate that there are no radiation components along directions transversal to $k^\mu_-$, see figure \ref{fig:doble}, in particular along the second null normal to $\Sc$, $k^\mu_+$. Observe that $\Q^\mu_- =0$ signifies that $k^\mu_-$ is a repeated principal null direction of the re-scaled Weyl tensor, and in this sense it may be thought of as the direction of propagation of asymptotic radiation. In turn, this signifies that $m^a$ is, on the given cut $\Sc$, an `incoming' direction that provides the direction of `evolution' of radiation at $\Sc$ within $\scri^+$ --in analogy with the null $n^a$ in the asymptotically flat case, figure \ref{fig:doble}. More importantly, as I am going to prove next, the condition $\Q^\mu_- =0$ can be expressed, in explicit manner, in terms of the triplet $(\scri^+,h_{ab},\F_{ab})$.
Assuming $Q_-^\alpha =0$ on ${\cal S}$ is equivalent, due to (\ref{Ws}), (\ref{Zs}) and (\ref{QAs}) for the minus sign, to 
\be
\F_A =\epsilon_{AB} C^B  \hspace{3mm} \mbox {and} \hspace{3mm} \hat{\F}_{AB}= \epsilon_{AD}\hat{C}_{B}{}^D .
\label{FC}
\ee
These conditions are actually stating that, on the cut $\Sc$
\be\label{DisC}
\fbox{$\F_{ab}-\frac{1}{2} \F_{cd}m^c m^d (3m_a m_b -h_{ab}) \eqc m^d\epsilon_{ed(a} \left(C_{b)}{}^e +m_{b)}m^f C_f{}^e \right).$}
\ee
This is our fundamental relation for cuts with only one radiation component. Note that this condition states that $\F_{ab}$ is determined by $C_{ab}$ (which is intrinsic to $(\scri, h_{ab})$) except for the one single component $\F_{cd}m^c m^d$, which is the only extra degree of freedom not given by the conformal geometry of $(\scri, h_{ab})$. This free degree of freedom concerns the Coulombian part of the gravitational field, proving that \eqref{DisC} certainly affects the radiative degrees of freedom. 

Using \eqref{DisC} one can readily compute the asymptotic super-Poynting vector on $\Sc$
$$
\left(\frac{3}{\Lambda}\right)^{3/2} \bar{p}^a\eqc -2m^a \left(C_{bc} C^{bc} + m^b C_{be} m_c C^{ce}\right) +4C^{ab} C_{bc} m^c +C_{bc} m^b m^c C^{ae} m_e -3(\F_{bc}m^b m^c)\epsilon^{ade}m_dC_{ef}m^f
$$
or equivalently (these can also be obtained from (\ref{pm}) and (\ref{pA}))
\bea
p_m =-2\left(\hat{\F}_{AB} \hat{\F}^{AB} +\F_A \F^A\right)=-2\left(\hat{C}_{AB} \hat{C}^{AB} +C_A C^A\right)\leq 0,\label{pm2} \\
p_A = \left[4\hat{\F}_{AB}+3\left(C^E{}_E \epsilon_{AB} - \F^E{}_E q_{AB} \right) \right] \F^B=
\left[4\hat{C}_{AB}+3\left(C^E{}_E q_{AB} - \F^E{}_E \epsilon_{AB} \right) \right] C^B . \label{pA2}
\eea

Concerning the asymptotic super-momentum $\Q^\alpha_+$, using again (\ref{Ws}), (\ref{Zs}) and (\ref{QAs}), now for the $+$ sign, one derives
$$
\W_+ =32 \hat{\F}_{AB} \hat{\F}^{AB}, \hspace{8mm} \Z_+ =16 \F_A \F^A ,  \hspace{8mm} Q^+_A = 32 \hat{\F}_{AB}\F^B.
$$
or equivalently
\bea
Q_+^\alpha =8 \left( 2\hat{\F}_{AB} \hat{\F}^{AB} k_-^\alpha + \F_A \F^A k^\alpha_+ +4 \hat{\F}^{AB}\F_B E^\alpha_A\right)\nonumber \\
=8 \left( 2\hat{C}_{AB} \hat{C}^{AB} k_-^\alpha + C_A C^A k^\alpha_+ +4 \hat{C}^{AB}C_B E^\alpha_A\right)\label{Q+}.
\eea
{\bf Remark}: It is remarkable that, with the restrictions put on $\F_{ab}$ in this case, $Q^\alpha_+$ is fully determined by the intrinsic geometry of $(\scri,h_{ab})$ and the cut $\Sc$ as follows from (\ref{Q+}). This is also true for $p_m$, see (\ref{pm2}). The only remaining `extrinsic' quantity identified above, $\F^E{}_E=\-\F_{ab}m^am^b$, only affects the components $p_A$ tangential to the cut. Another important point to remark is that $p_m=\bar{p}_a m^a \leq 0$ is non-positive, in accordance with our intuition that radiation in this situation travels towards the exterior of the cut $\Sc$ (figure \ref{fig:doble}), and provides an interesting interpretation for the balance law \eqref{balance}. Furthermore, $p_m=0$ implies that the entire $\bar{p}_a=0$ vanishes, and this statement again depends only on the intrinsic geometry of $(\scri,h_{ab})$ and the cut now.

If the discussed interpretation of the condition $\Q^\mu_- \eqc 0$ is to be accepted, then the absence of radiation determined by $\bar p_a$ should equivalently eliminate the unique radiative component that was left on the cut $\Sc$. This is proven in the following proposition.
\begin{Proposition}\label{prop:typeD}
The following conditions are all equivalent at any point of ${\cal S}$:
\begin{enumerate}
\item $Q_-^\mu =Q_+^\mu =0$.
\item $Q_-^\mu =0$ and $p_m=0$.
\item $Q_-^\mu =0$ and $\bar p_a=0$.
\item $\hat{\F}_{AB}=\hat{C}_{AB}=0$ and $\F_A =C_A=0$.
\item In the basis $\{\vec m,\vec E_A\}$
\be
(\F_{ab})= \F^E{}_E \left(
\begin{array}{ccc}
-1 & 0 & 0 \\
0 & 1/2 & 0 \\
0 & 0 & 1/2
\end{array}
 \right) ,
 \hspace{8mm}
 (C_{ab})= C^E{}_E \left(
\begin{array}{ccc}
-1 & 0 & 0 \\
0 & 1/2 & 0 \\
0 & 0 & 1/2
\end{array}
 \right) \label{typeD}
\ee
\end{enumerate}
\end{Proposition}
\begin{proof} 
I provide a circular proof $1\Rightarrow 2 \Rightarrow 3\Rightarrow 4 \Rightarrow 5 \Rightarrow 1$:
\begin{itemize}
\item If $Q_-^\mu =Q_+^\mu =0$ then from (\ref{Q+}) $\hat{C}_{AB}=0=C_A$ so that (\ref{pm2}) gives $p_m=0$.
\item If $Q_-^\mu =0$ and $p_m=0$, (\ref{pm2}) implies $\hat{C}_{AB}=0=C_A$ and together with (\ref{pA2}) gives that the full $\bar p_a$ vanishes.
\item If $Q_-^\mu =0$ and $\bar p_a=0$, (\ref{pm2}) implies $\hat{C}_{AB}=0=C_A$ and then (\ref{FC}) that also $\hat{\F}_{AB}=0=\F_A$. 
\item $\hat{\F}_{AB}=\hat{C}_{AB}=0$ and $\F_A =C_A=0$ is just saying that, in the mentioned basis, the matrices of $\F_{ab}$ and $C_{ab}$ take the form displayed in (\ref{typeD}).
\item If (\ref{typeD}) holds in the given basis, then $\hat{\F}_{AB}=\hat{C}_{AB}=0$ and $\F_A =C_A=0$ so that (\ref{Ws}--\ref{QAs}) imply $\W_\pm =\Z_\pm =0=Q_\pm^A$ and thus $Q_\pm^\mu =0$. $\Box$
\end{itemize}
\end{proof}
{\bf Remark}: This case corresponds to the situation where the rescaled Weyl tensor has Petrov type D at $\scri$ and is aligned at the cut $\Sc$, that is, the two multiple principal null directions are $\vec k_\pm$ (unless when also $\F^E{}_E=C^E{}_E=0$, but this corresponds to the de Sitter spacetime if $\scri \sim \mathbb{S}^3$).

Similar formulas and results are valid if one assumes $\Q^\mu_+ =0$ instead of $\Q^\mu_-=0$.

According to the nomenclature introduced in \cite{Fernandez-Alvarez_Senovilla-dS}, if on $\Delta\subset \scri$ there exists a foliation by cuts, all of them satisfying the property $\Q^\mu_-=0$, then we say that $\Delta$ is {\em strictly equipped and strongly oriented}, the vector field $m^a$ orthogonal to the cuts providing the orientation and equipment. If in addition the cuts are umbilical ($\S_{AB}=0$), $\Delta$ is both {\em strongly equipped and oriented} by $m^a$. The existence of news under such circumstances, as well as other possibilities, were explored at large in \cite{Fernandez-Alvarez_Senovilla-dS}. In particular we proved that the first component of news provides a good total News tensor field in the case of strongly equipped and oriented $\scri$.

\subsection{A conserved charge in vacuum}
As yet another justification for criterion \ref{crit2} let me present a conserved charge, built from the re-scaled Bel-Robinson tensor, that identifies the existence of radiation in asymptotic vacuum (this could be generalized to the case with matter) when the spacetime possesses conformal Killing vector fields. If the energy-momentum tensor of the physical spacetime vanishes in a neighbourhood ${\cal U}$ of $\scri^+$, then on that neighbourhood 
$$
\nabla_\rho {\cal D}^\rho{}_{\mu\nu\tau} \stackrel{{\cal U}}{=} 0.
$$
If $\xi^\mu_{{i}}$ are {\em any} three conformal Killing vectors on $(M,g)$ (they can be repeated) then the currents 
$$
{\cal B}^\rho (i,j,k):= \xi^\mu_{(i)} \xi^\nu_{(j)} \xi^\tau_{(k)}{\cal D}^\rho{}_{\mu\nu\tau}
$$
are divergence-free \cite{Senovilla2000,Lazkoz2003} on ${\cal U}$
$$
\nabla_\rho {\cal B}^\rho (i,j,k)\stackrel{{\cal U}}{=} 0.
$$
This implies that the `charges' defined on any spacelike hypersurface $\S$ without edge within ${\cal U}$ by
$$
{\cal B}_\S (i,j,k) := \int_\S {\cal B}^\rho(i,j,k) t_\rho
$$
(where $t_\rho$ is the unit normal to $\S$) are conserved, in the sense that they are independent of the choice of $\S$. In particular, they are equal to ${\cal B}_{\scri^+}(i,j,k)$.

If the $\xi^\mu_{(i)}=\xi^a_{(i)} e^\mu_a$ happen to be tangent to $\scri^+$, by using the explicit formulae in \cite{Alfonso2008} one can find (for instance, and for simplicity, for three copies of the same $\xi^\mu_{(1)}:=\xi^\mu$)
$$
{\cal B}_{\scri^+}(1,1,1) =\int_{\scri^+}  \left( \left(\frac{3}{\Lambda}\right)^{(3/2)} \bar{p}_a\xi^a-\xi_a\epsilon^{abc} \xi^d C_{bd} \xi^e \F_{ce} \right).
$$
This charge is generically non-zero. Nevertheless, if \eqref{DisC} holds and $\bar p_a=0$ then it vanishes. This is precisely the case of proposition \ref{prop:typeD}. This seems to hint in the direction that (non-zero) values of ${\cal B}_{\scri^+}(1,1,1)$ arise when there is gravitational radiation arriving at $\scri^+$.

\section{Symmetries with $\Lambda >0$}\label{sec:sym}
One of the missing elements to complete the picture in the $\Lambda >0$ scenario are the asymptotic symmetries. There is nothing like the BMS algebra/group and, the lack of a universal structure on $\scri$ is an impediment to provide a general notion of symmetries and, thereby, to look for appropriate conservation and balance laws. Still, one can try to find such missing symmetries in restricted situations, such as the one described in the previous section \ref{Q-=0} with strictly equipped and strongly oriented $\scri$, that is, if \eqref{DisC} holds on $\scri$.

To start with, let me argue that the `natural' definition for (infinitesimal) symmetries is any vector field $\vec Y\in\mathfrak{X}(\scri)$ leaving invariant the the tensor field 
$$
X_{abcdef}:= h_{ab} \F_{cd} \F_{ef}.$$
$X_{abcdef}$ is {\em gauge invariant} and contains the elements needed to determine any property of the physical spacetime, the triplet $(\scri,h_{ab},\F_{ab})$. Thus, a reasonable proposal of infinitesimal symmetries $\vec Y\in\mathfrak{X}(\scri)$ is simply
$$
\lied_Y (h_{ab} \F_{cd} \F_{ef})=0.
$$
This can be easily shown to be equivalent to
\be\label{basicsym}
\lied_Y h_{ab} =2 \psi h_{ab} , \hspace{1cm} \lied_Y D_{ab} = - \psi D_{ab}
\ee
for some function $\psi$. That this is a good definition is justified by noting that any solution $\vec Y$ of \eqref{basicsym} generates a Killing vector field on the physical spacetime and viceversa. This follows from a result due to Paetz \cite{Paetz2016}. Any solution of \eqref{basicsym} is termed {\em basic infinitesimal symmetry}. They satisfy
$$
\lied_Y \bar{p}_a =-5\psi \bar{p}_a .
$$

Nevertheless, an obvious problem arises with such basic symmetries. Observe that the first equation in \eqref{basicsym} informs us that $\vec Y$ must be a conformal Killing vector of $(\scri,h_{ab})$, and of course a generic 3-dimensional Riemannian manifold does not need to possess such vector fields. Hence, there are many $(\scri,h_{ab})$  without any {\em basic} infinitesimal symmetries. 

To remedy this situation, let me restrict the possible $(\scri,h_{ab})$ to those which possess a vector field $m^a$, orthogonal to a foliation of cuts, such that \eqref{DisC} holds on $\scri$, that is to say, $\scri$ is strictly equipped, and also strongly oriented,  by $m^a$. Then, we want that the symmetries preserve this structure, conformally keeping the orientation and equipment. This is achieved by the vector fields that satisfy 
\be\label{sym}
\lied_Y h_{ab}=2\psi h_{ab} +2\gamma m_a m_b, \hspace{1cm} \lied_Y m_a =(\gamma +\psi) m_a
\ee
for some functions $\psi$ and $\gamma$ on $\scri$. From this one also has
$$
\lied_Y m^a = -(\gamma +\psi) m^a .
$$
First of all, observe that the basic symmetries \eqref{basicsym} are included here (for $\gamma =0$) as long as they preserve the direction field $m^a$. Secondly, it is easy to check that the family of solutions of \eqref{sym} constitute a Lie algebra. Thirdly, the function $\gamma$ is gauge invariant under \eqref{gauge} while $\psi$ has the following behaviour
$$
\tilde\psi =\psi +\frac{1}{\omega} \lied_Y \omega.$$
Fourthly, equations \eqref{sym} are equivalent to
\be\label{sym2}
\lied_Y P_{ab} = 2\psi P_{ab} , \hspace{1cm} \lied_Y m_a =(\gamma +\psi) m_a
\ee
where 
$$
P_{ab} := h_{ab} -m_a m_b
$$
is the orthogonal projector of the foliation defined by $m_a$ that projects to the leaves. In this form, and given that the projector restricted to each leaf $\Sc$ of the foliation gives the corresponding first fundamental form $q_{AB}$, the first relation in \eqref{sym2} states that the vector fields leave the conformal metrics invariant. Actually, \eqref{sym2} or \eqref{sym} are an example of the infinitesimal symmetries called bi-conformal vector fields \cite{GarciaParrado2004} that leave two orthogonal distributions  conformally invariant.  As proved in  \cite{GarciaParrado2004}, the solutions of \eqref{sym2} can form an infinite-dimensional Lie algebra. 

It remains the question of whether or not these new symmetries can be somehow derived as asymptotic generalized symmetries from the physical spacetime. This is certainly the case, as we briefly explain next. Start by considering a vector field $\hat\xi^\mu$ on the physical spacetime $(\hat M,\hat g)$ such that is has a smooth extension to $\scri$ on $M$. Then on $\hat M$
$$
\lied_{\hat\xi} g_{\mu\nu} = \Omega^2 \lied_{\hat\xi}\hat{g}_{\mu\nu}+ \frac{2}{\Omega} \lied_{\hat\xi}\Omega \hat{g}_{\mu\nu}
$$
and require that 
$$
H_{\mu\nu} := \Omega^2 \lied_{\hat\xi}\hat{g}_{\mu\nu}
$$
has a regular limit to $\scri$. The basic idea is to find the `minimum' possible $H_{\mu\nu}$ that induces the symmetries on $(\scri,h_{ab})$. In other words, $\hat\xi^\mu$ can be thought as an approximate symmetry when approaching infinity. One can easily prove \cite{Fernandez-Alvarez_Senovilla-dS} that 
$$
\xi^\mu n_\mu \eqs 0 \Longrightarrow \, \, \hat\xi^\mu =Y^a e^\mu_a
$$
and $Y^a$ is a vector field on $\scri$. It is necessary to take into account that only the class of vector fields $\hat\xi^\mu$ defined modulo the addition of any term of the form $\Omega v^\mu$, for arbitrary $v^\mu$, makes sense. This implies that combinations of type
$$
v_\mu n_\nu + v_\nu n_\mu -2v^\rho n_\rho g_{\mu\nu} + 2\Omega (\nabla_\mu v_\nu +\nabla_\nu v_\mu)
$$
can be added to $H_{\mu\nu}$ without changing the sought asymptotic symmetry. 

Thus, in order to choose $H_{\mu\nu}$ one first notices that $H_{\mu\nu}\propto g_{\mu\nu}$ (including $H_{\mu\nu}=0$, which mimics the case of $\Lambda =0$ as studied in \cite{Geroch81}) will lead to conformal Killing vectors of $(\scri,h_{ab})$, that is, to the basic symmetries \eqref{basicsym}. Thus, one needs a more general choice. The next `minimal' possible such choice is that $H_{\mu\nu}$ is a rank-1 tensor field on $\scri$, that is, there exists a vector field $m^\mu$ such that $H_{\mu\nu} = F m_\mu m_\nu$, or including the redundant terms above
$$
H_{\mu\nu} = F m_\mu m_\nu +v_\mu n_\nu + v_\nu n_\mu -2v^\rho n_\rho g_{\mu\nu} + 2\Omega (\nabla_\mu v_\nu +\nabla_\nu v_\mu)
$$
where necessarily $m^\mu n_\mu \eqs 0$ \cite{Fernandez-Alvarez_Senovilla-dS}. Projection to $\scri$ then shows that \cite{Fernandez-Alvarez_Senovilla-dS}
$$
\lied_Y h_{ab}= 2\psi h_{ab} +2\gamma m_a m_b 
$$
where $\gamma =F|_\scri$ and $\psi = -(2v^\rho n_\rho +\xi^\mu n_\mu/\Omega)|_\scri$. This is precisely the first in \eqref{sym}, and the Lie algebra property requires the second one. 

The precise structure of the Lie algebra of the symmetries \eqref{sym} depends on the specific situation, that is, on the particular properties of the foliation determined by the vector field $m^a$ that equips and orientates $\scri$. For instance, in the case that the orientation and the equipment are both strong (so that the foliation is by umbilical cuts), the structure is the product of conformal transformations of the cuts times an ideal which commutes with the previous and depends on arbitrary functions, so that the algebra is infinite dimensional  \cite{Fernandez-Alvarez_Senovilla-dS}.

\section{Closing comments with examples}\label{sec:fin}
Criteria \ref{crit1} and \ref{crit2} have been tested in a variety of spacetimes \cite{Fernandez-Alvarez_Senovilla20b,Fernandez-Alvarez_Senovilla-dS} that admit a conformal completion and so far they agree with the expected results concerning existence of gravitational radiation, as well as in relation to other concepts introduced in this paper. Herein I provide a summary of the known results and add a couple of new ones.

First of all, take spherically symmetric spacetimes that, as we know, do not contain any kind of gravitational radiation. If they admit a conformal completion this can be assumed to have spherical symmetry too, and then $C_{ab}$ and $\F_{ab}$ inherit the symmetry. This readily proves that $C_{ab}$ and $\F_{ab}$ must be proportional to each other, so that their commutator vanishes and using \eqref{superp} this leads to $\bar{p}_a=0$ in agreement with the absence of radiation in such situations according to our criteria. This includes, in particular, de Sitter spacetime which actually has both $C_{ab}$ and $\F_{ab}$ vanishing, where one can identify the 10 asymptotically basic infinitesimal symmetries, four possible strong equipments (all of them equivalent) with umbilical foliations by $\mathbb{S}^2$ cuts, and find the structure of the group of symmetries of type \eqref{sym} for any of the strong equipments. This is composed of the conformal Killing vectors of the sphere together with a vector field of type $F(\chi) \partial_\chi$, for {\em arbitrary} function $F$, where $\chi$ is a typical latitud coordinate on the 3-dimensional sphere \cite{Fernandez-Alvarez_Senovilla-dS}.

Next, consider the ``Kerr-de Sitter-like spacetimes'' as defined in \cite{Mars2016}. Basically, these are the $\Lambda $-vacuum spacetimes with a Killing vector filed whose `Mars-Simon' tensor vanishes \cite{MS2015} and admit  a conformal completion. They include in particular the Kerr-de Sitter solution as well as many others \cite{MS2015,Mars2016,Mars2017,MPS1}.
Kerr-de Sitter-like spacetimes are characterized by initial data $(\scri,h_{ab},\F_{ab})$ with
$$
C_{ab}= \frac{A}{|Y|^5} \left(Y_a Y_b -\frac{1}{3} |Y|^2 h_{ab} \right), \hspace{3mm} 
\F_{ab}= \frac{B}{|Y|^5} \left(Y_a Y_b -\frac{1}{3} |Y|^2 h_{ab} \right)
$$
for some constants $A,B$ where $Y^a\in \mathfrak{X}(\S)$ is a conformal Killing vector on $(\scri,h)$ with no fixed points. $Y^a$ is the conformal Killing vector induced by the Killing vector of the physical spacetime with vanishing Mars-Simon tensor. From the expressions above we check that again $C_{ab}$ and $\F_{ab}$ are proportional to each other so that \eqref{superp} imples $\bar p^a=0$ and criterion \ref{crit2} states that there is no gravitational radiation. This is also an expected result.
In the particular case of Kerr-de Sitter spacetime, including the Kottler solution for zero angular momentum, the constant $A=0$ ($(\scri,h_{ab})$ is conformally flat), there are two strong orientations but neither of them leads to a strong equipment. The corresponding symmetries \eqref{sym} coincide with the basic asymptotic symmetries \eqref{basicsym} and are induced by the two Killing vectors of the spacetime. Still, there exists a `natural' strong equipment by umbilical cuts and the corresponding algebra of symmetries \eqref{sym} is again infinite dimensional depending on an arbitrary function of one variable \cite{Fernandez-Alvarez_Senovilla-dS}.

%To gain some intuition about in- and out-going super.momenta, one can check their behaviour in Kerr-de Sitter like spacetimes.
%$\Q_\pm^\alpha$ are actually non-zero for Kerr-dS-like spacetimes in general, but they do always satisfy ({\color{red} for any cut})
%$$ \Z_+ =\Z_-, \hspace{1cm} \W_+ =\W_-, \hspace{1cm} \overline{\Q}_+^A+\overline{\Q}_-^A =0$$
%This implies {\color{blue} $\Q_+^\alpha + \Q_-^\alpha = \left(\W_+ +\Z_+\right)n^\alpha,$}
%Equivalently {\color{red} $\overline{\Q}_+^a+\overline{\Q}_-^a=0$}.
%If $m^a$ is chosen parallel to $Y^a$ (when this is feasible, e.g. Kottler) then actually both $\Q_\pm^\alpha=0$ {\color{blue} They obviously comply with both proposed conditions!}

In \cite{Mars2016} a more general class of spacetimes, termed {\em asymptotically Kerr-de Sitter-like spacetimes}, was introduced. They also have a Killing vector but now the Mars-Simon tensor is only required to vanish asymptotically. Their characterization at infinity is given by data $(\scri,h_{ab},\F_{ab})$ such that
$$
C_{ab} Y^b = \delta Y_a, \hspace{1cm} \F_{ab} Y^b =\beta Y_a
$$
for some functions $\delta,\beta$ on $\scri$, where $Y^a$ is the conformal Killing vector on $(\scri,h_{ab})$ induced by the Killing vector of the physical spacetime. In other words, $C_{ab}$ and $\F_{ab}$ have $Y^a$ as a common eigenvector field. Obviously, the Kerr-de Sitter-like spacetimes are included here, but there are many other possibilities. In this case, gravitational radiation may be present. An interesting possibility is the analysis of asymptotically Kerr-de Sitter-like spacetimes which also comply with \eqref{DisC} for some $m^a$. If in this case $Y^a$ points into the direction $m^ a$ that equips $\scri$, that is to say, $Y^a= |Y| m^a$ then the eigenvalues of the common eigendirection are
$$
\delta = C_{ab}m^am^b, \hspace{1cm} \beta =\F_{ab}m^am^b
$$
and  also $\F_A=0$ and $C_A=0$. Equation \eqref{pA2} tells us that $p_A=0$ and thus from \eqref{pm2}
$$
\left(\frac{3}{\Lambda}\right)^{3/2} \bar{p}_a = -\hat{C}_{AB} \hat{C}^{AB} m_a.
$$

Next, a very interesting spacetime to be used as example is the $C$-metric \cite{Stephani2003,Griffiths-Podolsky2009}, both in the $\Lambda>0$ and $\Lambda =0$ cases, see \cite{Fernandez-Alvarez_Senovilla20b,Fernandez-Alvarez_Senovilla-dS}, because this is known to have gravitational radiation in the asymptotically flat case \cite{Ashtekar-Dray81}.  The existence of gravitational radiation according to our criterion \ref{crit2} for $\Lambda\geq 0$ was proven in \cite{Fernandez-Alvarez_Senovilla20b}. For the $C$-metric there are two possible strong orientations, both of them providing strong equipments, and the Lie algebra of symmetries \eqref{sym} is infinite dimensional once more, but in this case depending of multiple arbitrary functions \cite{Fernandez-Alvarez_Senovilla-dS}.

Another interesting family of spacetimes usable as examples are the Robinson-Trautman metrics \cite{Stephani2003,Griffiths-Podolsky2009}, for $\Lambda \geq 0$. Generically, they have one strong orientation which defines a strong equipment, and the corresponding asymptotic symmetries \eqref{sym} form an infinite-dimensional Lie algebra that depends on an arbitrary function of one variable. They generically contain gravitational radiation according to criterion \ref{crit2}, the particular case of Petrov type N Robinson-Trautman metrics is analyzed in detail in \cite{Fernandez-Alvarez_Senovilla-dS}.

%%%%%%%%%%%%%%%%%%%%%%%%%%%%%%%%%%%%%%%%%%
%\section{Patents}

%This section is not mandatory, but may be added if there are patents resulting from the work reported in this manuscript.

%%%%%%%%%%%%%%%%%%%%%%%%%%%%%%%%%%%%%%%%%%
\vspace{6pt} 

%%%%%%%%%%%%%%%%%%%%%%%%%%%%%%%%%%%%%%%%%%
%% optional
%\supplementary{The following are available online at \linksupplementary{s1}, Figure S1: title, Table S1: title, Video S1: title.}


%%%%%%%%%%%%%%%%%%%%%%%%%%%%%%%%%%%%%%%%%%
%\authorcontributions{For research articles with several authors, a short paragraph specifying their individual contributions must be provided. The following statements should be used ``Conceptualization, X.X. and Y.Y.; methodology, X.X.; software, X.X.; validation, X.X., Y.Y. and Z.Z.; formal analysis, X.X.; investigation, X.X.; resources, X.X.; data curation, X.X.; writing---original draft preparation, X.X.; writing---review and editing, X.X.; visualization, X.X.; supervision, X.X.; project administration, X.X.; funding acquisition, Y.Y. All authors have read and agreed to the published version of the manuscript.'', please turn to the  \href{http://img.mdpi.org/data/contributor-role-instruction.pdf}{CRediT taxonomy} for the term explanation. Authorship must be limited to those who have contributed substantially to the work~reported.}

\section*{Funding}
Research supported by Basque Government grant numbers IT956-16 and IT1628-22, and by Grant FIS2017-85076-P funded by the Spanish MCIN/AEI/10.13039/501100011033 and by "ERDF A way of making Europe".

\section*{Acknowledgments} 
Discussions with Fran Fern\'andez-\' Alvarez are gratefully acknowledged.

%\conflictsofinterest{The author declares no conflict of interest.%Declare conflicts of interest or state ``The authors declare no conflict of interest.'' Authors must identify and declare any personal circumstances or interest that may be perceived as inappropriately influencing the representation or interpretation of reported research results. Any role of the funders in the design of the study; in the collection, analyses or interpretation of data; in the writing of the manuscript, or in the decision to publish the results must be declared in this section. If there is no role, please state ``The funders had no role in the design of the study; in the collection, analyses, or interpretation of data; in the writing of the manuscript, or in the decision to publish the~results''.} 

%% Optional
%\sampleavailability{Samples of the compounds ... are available from the authors.}

%%%%%%%%%%%%%%%%%%%%%%%%%%%%%%%%%%%%%%%%%%
%% Only for journal Encyclopedia
%\entrylink{The Link to this entry published on the encyclopedia platform.}

%%%%%%%%%%%%%%%%%%%%%%%%%%%%%%%%%%%%%%%%%%
%% Optional
%%%%%%%%%%%%%%%%%%%%%%%%%%%%%%%%%%%%%%%%%%
%% Optional
%\appendixtitles{yes} % Leave argument "no" if all appendix headings stay EMPTY (then no dot is printed after "Appendix A"). If the appendix sections contain a heading then change the argument to "yes".
%\appendixstart
\appendix

\section[\appendixname~\thesection]{`(Super)-energy' tensors in a nutshell}\label{App:1}
%All appendix sections must be cited in the main text. In the appendices, Figures, Tables, etc. should be labeled, starting with ``A''---e.g., Figure A1, Figure A2, etc.
Given any tensor (field), say $t_{\mu_1\dots\mu_m}$, there is a canonical way \cite{Senovilla2000} of constructing a new tensor (field) $T\{t\}_{\mu_1\dots \mu_{2s}}$ quadratic on $t_{\mu_1\dots\mu_m}$ and satisfying the {\em dominant property}, that is to say
\be
T\{t\}_{\mu_1\dots \mu_{2s}}u^{\mu_1} \dots v^{\mu_{2s}}\geq 0
\ee
for arbitrary future-pointing vectors $u^{\mu_1} \dots v^{\mu_{2s}}$. The inequality is strict if all the vectors $u^{\mu_1} \dots v^{\mu_{2s}}$ are timelike. In particular, the total timelike component in an orthonormal basis $\{\vec{e}_{\alpha}\}$ whose timelike direction is given by $\vec{e}_0$, that is,
$$
T_{0\dots 0} := T\{t\}_{\mu_1\dots \mu_{2s}}e_0^{\mu_1} \dots e_0^{\mu_{2s}}\geq 0
$$
is positive and vanishes if and only if $t_{\mu_1\dots\mu_m}=0$. Such quadratic tensors are called `super-energy' tensors generically, and its total timelike component is the `super-energy' of $t_{\mu_1\dots\mu_m}$ relative to the chosen $\vec{e}_0$. The fully symmetric part $T\{t\}_{(\mu_1\dots \mu_{2s})}$ ---which is the only part relevant for the super-energy of $t_{\mu_1\dots\mu_m}$--- is {\em unique} with the above properties. 

If the underlying, seed, tensor $t_{\mu_1\dots\mu_m}$ is actually a $p$-form, then $s=1$ and $T\{t\}_{\mu\nu}$ is a rank-2 symmetric tensor. In particular, if $t_\mu =\nabla_\mu \phi$ is an exact one-form, then $T\{\nabla\phi\}_{\mu\nu}$ is the standard energy-momentum tensor of a massless scalar field $\phi$; while if $t_{\mu\nu}=F_{[\mu\nu]}$ is a 2-form, then $T\{F\}_{\mu\nu}$ is the standard energy-momentum tensor of the electromagnetic field $F_{\mu\nu}$. For further details, see \cite{Senovilla2000}.

In this article, we are interested in the super-energy tensor $T\{W\}$ of {\em Weyl-tensor candidates} $W_{\alpha\beta\mu\nu}$. A Weyl tensor candidate is a double (2,2)-form with the same symmetry and trace properties of the Weyl tensor:
$$
W_{\alpha\beta\mu\nu}=W_{[\alpha\beta][\mu\nu]}, \hspace{1cm} W_{\alpha[\beta\mu\nu]}=0, \hspace{1cm} W^\rho{}_{\beta\rho\mu}=0 .
$$
Its super-energy tensor is the rank-4 tensor
\bean
T\{W\}_{\alpha\beta\lambda\mu} &=& W_{\alpha\rho\lambda\sigma} W_\beta{}^\rho{}_\mu{}^\sigma + W_{\alpha\rho\mu\sigma} W_\beta{}^\rho{}_\lambda{}^\sigma
-\frac{1}{2} g_{\alpha\beta} W_{\tau\rho\lambda\sigma} W^{\tau\rho}{}_\mu{}^\sigma\\
&&-\frac{1}{2} g_{\lambda\mu} W_{\alpha\rho\tau\sigma} W_\beta{}^{\rho\tau\sigma}
+\frac{1}{8} g_{\alpha\beta}g_{\lambda\mu} W_{\nu\rho\tau\sigma} W^{\nu\rho\tau\sigma}
\eean
which, in 4-dimensional spacetime reduces to simply 
\be\label{BRW}
T\{W\}_{\alpha\beta\lambda\mu} = W_{\alpha\rho\lambda\sigma} W_\beta{}^\rho{}_\mu{}^\sigma + W_{\alpha\rho\mu\sigma} W_\beta{}^\rho{}_\lambda{}^\sigma
-\frac{1}{8} g_{\alpha\beta}g_{\lambda\mu} W_{\nu\rho\tau\sigma} W^{\nu\rho\tau\sigma}.
\ee
This tensor is fully symmetric and traceless \cite{Senovilla2000,Senovilla97}. It also admits the alternative expression (still in 4 dimensions)
\be\label{BRW1}
T\{W\}_{\alpha\beta\lambda\mu} = W_{\alpha\rho\lambda\sigma} W_\beta{}^\rho{}_\mu{}^\sigma + \stackrel{*}{W}_{\alpha\rho\lambda\sigma} \stackrel{*}{W}_\beta{}^\rho{}_\mu{}^\sigma 
\ee
where 
$$
\stackrel{*}{W}_{\alpha\rho\lambda\sigma}:= \frac{1}{2} \eta_{\alpha\rho\mu\nu} W^{\mu\nu}{}_{\lambda\sigma} 
$$
and $\eta_{\alpha\rho\mu\nu}$ is the canonical volume element 4-form. 

If the Weyl-tensor candidate is divergence-free, $\nabla_\rho W^\rho{}_{\beta\mu\nu}=0$, then $T\{W\}_{\alpha\beta\lambda\mu} $ is divergence-free too. 

When $W_{\alpha\beta\mu\nu}=C_{\alpha\beta\mu\nu}$ is the true Weyl tensor, $T\{C\}_{\alpha\beta\lambda\mu} $ is called the Bel-Robinson tensor \cite{Senovilla97,Bel1958,Bel1962}.

\section[\appendixname~\thesection]{The tensor $\rho_{AB}$ for conformal classes of 2-dimensional Riemannian manfiolds}\label{App:rho}
In this appendix an important tensor field available in 2-dimensional Riemannian manifolds with relevant conformal properties is presented. This tensor is reminiscent of another one introduced by Geroch for $\scri$ in an asymptotically flat situation \cite{Geroch1977} and allows one to extract the news tensor field from the pullback of the Schouten tensor $S_{ab}$, as explained in section \ref{sec:Lambda=0}. The invariant interpretation and significance of this tensor field is discussed in this Appendix, see also \cite{Fernandez-Alvarez_Senovilla-dS}. 
			
			As all possible 2-dimensional Riemannian manifolds are (locally) conformal to the round sphere, let us start by considering the round sphere $(\mathbb{S}^2,q_{round})$ with constant Gaussian curvature $K$ , given in conformally flat form in Cartesian coordinates $\{x,y\}$ by 
$$
q_{round} = \left[1+\frac{K}{4} (x^2+y^2) \right]^{-2} \left(dx^2 + dy^2 \right).
$$
Using canonical angular coordinates on $\mathbb{S}^2$ via the standard stereographic projection from the north pole
\bean
x=\frac{2}{\sqrt{K}} \cot \frac{\theta}{2} \cos\varphi , \hspace{1cm} y=\frac{2}{\sqrt{K}} \cot \frac{\theta}{2} \sin\varphi ,\\
\theta =2\arctan \frac{2}{\sqrt{K(x^2+y^2)}} , \hspace{1cm} \varphi =\arctan \frac{y}{x}
\eean
with $\theta\in (0,\pi]$ and $\varphi\in [0,2\pi)$, the metric becomes
\be
q_{round}=\frac{1}{K} \left(d\theta^2 +\sin^2\theta d\varphi^2 \right)\label{round}
\ee
and the part in parenthesis is the metric of the unit round sphere, which will be denoted in index notation by $\Omega_{AB}$ from now on. As is well known, the sphere possesses a 6-dimensional algebra of {\em global} conformal Killing vector fields ---see e.g. Appendix F in \cite{Fernandez-Alvarez_Senovilla-dS}---, an appropriate basis for them is
\bea
\vec{\xi}_1  &=& -\left(\sin\varphi \partial_\theta +\cot\theta \cos\varphi\partial_\varphi \right)\label{Kil1}\\
\vec\xi_2 &=& \cos\varphi \partial_\theta -\cot\theta \sin\varphi\partial_\varphi ,\label{Kil2}\\
\vec{\xi}_3 &=&\partial_\varphi,\label{Kil3}\\
\vec\eta_1 &=& \cos\theta\cos\varphi\partial_\theta -\frac{\sin\varphi}{\sin\theta}\partial_\varphi , \label{CKil1}\\
\vec\eta_2 &=& \cos\theta\sin\varphi\partial_\theta +\frac{\cos\varphi}{\sin\theta}\partial_\varphi \label{CKil2}\\
 \vec\eta_3 &=& -\sin\theta \partial_\theta \label{CKil3} .
\eea
The first three are actually Killing vectors generating the group SO(3) while the remaining three are proper conformal Killing vectors satisfying ($i=1,2,3$, $D_A$ is the covariant derivative on the sphere)
$$
D_A \eta^B_{(i)} =-\delta^B_A n_{(i)} 
$$
where 
$$
n_{(i)} =\left(\sin\theta\cos\varphi,\sin\theta\sin\varphi, \cos\theta\right).%=-\left(\Psi_{\vec\eta_1},\Psi_{\vec\eta_2}, \Psi_{\vec\eta_3} \right)
$$

Observe that the three CKVs (\ref{CKil1}-\ref{CKil3}) are all exact one-forms%(and on the simply connected $\mathbb{S}^2$ actually exact):
$$
\bm{\eta}_1 = \frac{1}{K} d(\sin\theta\cos\varphi), \hspace{6mm} \bm{\eta}_2= \frac{1}{K} d(\sin\theta\sin\varphi),  \hspace{6mm}  \bm{\eta}_3 =\frac{1}{K} d(\cos\theta),
$$
or more compactly
$$
\eta_{(i)B} = \frac{1}{K} D_B n_{(i)}
$$
while the three Killing vector fields (\ref{Kil1}-\ref{Kil3}) are co-exact
$$
\xi^A_{(i)} = \epsilon^{AB}\eta_{(i)B} =\frac{1}{K} D_B (\epsilon^{AB} n_{(i)})
$$
where $\epsilon_{AB}$ is the volume 2-form. This leads to the known result 
\be\label{DDn}
D_A D_B n_{(i)} =-\Omega_{AB} n_{(i)} .
\ee
Notice that in particular $\Delta n_{(i)} = -2K n_{(i)}$, where $\Delta$ is the Laplacian on the sphere, meaning that $n_{(i)}$ are the three spherical harmonics $Y_1^i$, with $l=1$. These three, together the spherical harmonic of order $l=0$, can thus be combined into a single covariant `4-vector' 
$$
\pi_{(\mu)} := (1,n_{(i)}) 
$$
which is null in an auxiliary Minkowski metric: $\eta^{\mu\nu} \pi_{(\mu)}\pi_{(\nu)}=0$. Using (\ref{DDn}) one can then write (here each $\pi_{(\mu)}$ is considered as a function)
\be\label{DDpi}
D_A D_B \pi_{(\mu)} - \frac{1}{2} \Delta \pi_{(\mu)} \frac{1}{K}\Omega_{AB} =0.
\ee

The question that arises is:  Is there a conformally invariant version of (\ref{DDpi}), valid in arbitrary 2-dimensional Riemannian manifolds with metric $q_{AB}$? To answer this question, perform a general conformal transformation 
$$
\tilde{q}_{AB}=\omega^2 q_{AB}
$$
and assume that the four $\pi_{(\mu)}$ transform in a ``coordinated'' and homogeneous manner so that
$$
\tilde{\pi}_{(\mu)} = H(\omega) \pi_{(\mu)}
$$
for some function $H(\omega)$ to be determined. A direct calculation using the change of the covariant derivative under conformal re-scalings leads then to
\bea
\tilde{D}_A \tilde{D}_B \tilde{\pi}_{(\mu)} =HD_A D_B \pi_{(\mu)} +D_A H D_B \pi_{(\mu)} +D_B H D_A \pi_{(\mu)}+\pi_{(\mu)} D_A D_B H\nonumber \\
-\frac{1}{\omega} \left[D_A\omega D_B(H\pi_{(\mu)}) + D_B \omega D_A (H\pi_{(\mu)}) -q^{CE}D_C \omega D_E(H\pi_{(\mu)}) q_{AB}\right]\label{DDpi1}
\eea
whose trace reads
\be\label{trDDpi1}
\tilde\Delta \tilde \pi_{(\mu)}=\frac{1}{\omega^2} \left(H\Delta \pi_{(\mu)} +2 q^{CE} D_C \omega D_E\pi_{(\mu)} +\pi_{(\mu)}\Delta H \right)
\ee
so that the combination of \eqref{DDpi1} and \eqref{trDDpi1} produces
\bea
\tilde{D}_A \tilde{D}_B \tilde{\pi}_{(\mu)} -\frac{1}{2} \tilde{q}_{AB} \tilde\Delta \tilde{\pi}_{(\mu)} =H\left(D_A D_B \pi_{(\mu)} - \frac{1}{2} \Delta \pi_{(\mu)} q_{AB}  \right)\nonumber\\
+ \pi_{(\mu)} \left(D_A D_B H -\frac{1}{\omega} D_A\omega D_B H -\frac{1}{\omega} D_B H D_A\omega -\frac{1}{2} \Delta H q_{AB} +\frac{1}{\omega} q^{CE} D_c \omega D_E H  q_{AB}\right)\nonumber\\
+\omega \left[ D_A \pi_{(\mu)} D_B\left(\frac{H}{\omega}\right) +D_B \pi_{(\mu)} D_A\left(\frac{H}{\omega}\right)- q^{CE}D_C\pi_{(\mu)} D_E\left(\frac{H}{\omega} \right) q_{AB}\right].\label{DDpi2}
\eea
Hence, the only way that this can lead to a conformally well-behaved relation is that the terms with $D_A \pi_{(\mu)}$ dissapear, which requires
$$
H = \omega
$$
where an arbitrary multiplicative constant has been set to $1$ by a simple redefinition of $\pi_{(\mu)}$. Introducing this into \eqref{DDpi2} one gets
\bea
\tilde{D}_A \tilde{D}_B \tilde{\pi}_{(\mu)} -\frac{1}{2} \tilde{q}_{AB} \tilde\Delta \tilde{\pi}_{(\mu)} =\omega\left(D_A D_B \pi_{(\mu)} - \frac{1}{2} \Delta \pi_{(\mu)} q_{AB}  \right)\nonumber\\
+ \pi_{(\mu)} \left(D_A D_B \omega -\frac{2}{\omega} D_A\omega D_B \omega -\frac{1}{2} \Delta \omega q_{AB} +\frac{1}{\omega} q^{CE} D_c \omega D_E \omega  q_{AB}\right)\label{DDpi3}
\eea
To make sense of the conformal behaviour of this expression notice that the first line contains the same combination on both sides and thus the second line must go partly to one side and partly to the other side in a concordant manner. The terms multiplying $q_{AB}$ can be easily rearranged by using the relation between Gaussian curvatures of conformally related metrics
\be\label{K'}
\tilde{K} =\frac{1}{\omega^2} \left(K -\frac{1}{\omega} \Delta \omega +\frac{1}{\omega^2} q^{CB} \omega_B \omega_C \right)=
\frac{1}{\omega^2} \left(K- \Delta \ln \omega \right).
\ee
where I use the notation $\omega_A := D_A \omega$. Then \eqref{DDpi3} becomes
\bea
\tilde{D}_A \tilde{D}_B \tilde{\pi}_{(\mu)} -\frac{1}{2} \tilde{q}_{AB} \tilde\Delta \tilde{\pi}_{(\mu)}-\frac{\tilde K}{2} \tilde{q}_{AB} \tilde{\pi}_{(\mu)}=
\omega\left(D_A D_B \pi_{(\mu)} - \frac{1}{2} \Delta \pi_{(\mu)} q_{AB} -\frac{K}{2} q_{AB} \pi_{(\mu)}  \right)\nonumber\\
+ \pi_{(\mu)} \left(D_A D_B \omega -\frac{2}{\omega} D_A\omega D_B \omega +\frac{1}{2\omega} q^{CE} D_c \omega D_E \omega  q_{AB}\right) . \label{DDpi4}
\eea
If our goal is achievable, the second line here must be the difference between a symmetric tensor field and its tilded version --up to a factor $\omega$. Call this tensor field $\rho_{AB}$, and set
$$
\rho_{AB} -\tilde{\rho}_{AB} :=\frac{1}{\omega} D_A D_B \omega -\frac{2}{\omega^2} D_A\omega D_B \omega +\frac{1}{2\omega^2} q^{CE} D_c \omega D_E \omega  q_{AB}
$$
which renders \eqref{DDpi4} in the form
\bean
\tilde{D}_A \tilde{D}_B \tilde{\pi}_{(\mu)} -\frac{1}{2} \tilde{q}_{AB} \tilde\Delta \tilde{\pi}_{(\mu)}+\left(\tilde{\rho}_{AB}-\frac{\tilde K}{2} \tilde{q}_{AB}\right) \tilde{\pi}_{(\mu)}\\
=\omega\left[D_A D_B \pi_{(\mu)} - \frac{1}{2} \Delta \pi_{(\mu)} q_{AB}+\left(\rho_{AB} -\frac{K}{2} q_{AB} \right)\pi_{(\mu)}  \right]
\eean
This is the sought result, providing the right expression which is well behaved and answers in the affirmative our question. Hence the equation valid in arbitrary metrics $q_{AB}$ on the sphere reads (with $D_A$ the covariant derivative for $q_{AB}$ and $\Delta$ and $K$ the corresponding Laplacian and Gaussian curvature, respectively)
\begin{equation}\label{eq:Hess-gen}
			D_A D_B \pi_{(\mu)} -\frac{1}{2} q_{AB} \Delta\pi_{(\mu)}+\left( \rho_{AB} -\frac{K}{2} q_{AB}\right)\pi_{(\mu)}=0
			\end{equation} 
{\em as long as the tensor field $\rho_{AB}$ behaves, under conformal re-scalings of type (\ref{gaugeh}), like}
			\begin{equation}\label{eq:rho-rho'}
			 \tilde{\rho}_{AB} =\rho_{AB}- \frac{1}{\omega}D_A D_B \omega +\frac{2}{\omega^2} D_A\omega D_B\omega -\frac{1}{2\omega^2}q_{AB}  q^{CD} D_C\omega D_D\omega .
			\end{equation} 
If this holds, and if $\pi_{(\mu)}$ are the four solutions of (\ref{eq:Hess-gen}), then $\tilde{\pi}_{(\mu)}=\omega \pi_{(\mu)}$ are the corresponding four solutions in the re-scaled metric $\tilde{q}_{AB}=\omega^2 q_{AB}$. Notice that the constraint $\eta^{\mu\nu} \pi_{(\mu)}\pi_{(\nu)}=0$ with the auxiliary Minkowski metric remains invariant. 
			
The trace of \eqref{eq:Hess-gen} leads to 
			\begin{equation}\label{eq:traceOFrho}
			q^{AB} \rho_{AB}= K
			\end{equation}
which, taking (\ref{eq:rho-rho'}) into account, also holds in any gauge because of \eqref{K'}.

Observe that, if we wish to recover (\ref{DDpi}) in the round gauge, (\ref{eq:Hess-gen}) requires that in that gauge $\rho_{AB} =(K/2) q_{AB}=(1/2)\Omega_{AB}$ so that $D_C \rho_{AB}=0$ holds in that round gauge. In particular, 
\begin{equation}\label{eq:Derivadarho} 
			D_{[C} \rho_{A]B} =0
\end{equation}
and this formula holds in any gauge due to \eqref{eq:rho-rho'} and \eqref{eq:traceOFrho}. Properties \eqref{eq:rho-rho'} and \eqref{eq:Derivadarho} uniquely determine the tensor $\rho_{AB}$ if the 2-dimensional manifold has topology $\mathbb{S}^2$ (Corollary \ref{coroRho} below) or, more generally, for arbitrary topology if there is a conformal Killing vector with a fixed point. 
This follows from the following set of results.
\begin{Lemma}\label{lem:dt2}
Let $(\Sc,q_{AB})$ by any 2-dimensional Riemannian manifold and $t_{AB}=t_{(AB)}$ any symmetric tensor field on ${\cal S}$ whose gauge behaviour under residual gauge transformations (\ref{gaugeh}) is
\be\label{normalgauge}
\tilde{t}_{AB}= t_{AB}-\frac{a}{\omega}D_A\omega_B +\frac{2a}{\omega^2} \omega_A\omega_B -\frac{a}{2\omega^2} \omega^D\omega_D q_{AB}
\ee
for some fixed constant $a\in \mathbb{R}$. Then, 
\be
\tilde{D}_{[C}\tilde{t}_{A]B}=D_{[C}t_{A]B}+\frac{1}{\omega} (aK-t^E{}_E) \omega_{[C}q_{A]B}.
\label{Dtgauge2}
\ee
In particular, if $n_{AB}=n_{(AB)}$ is any symmetric and gauge-invariant tensor field on ${\cal S}$, then,
\be
\tilde{D}_{[C}\tilde{n}_{A]B}=D_{[C}n_{A]B}-\frac{1}{\omega} n^E{}_E \omega_{[C}q_{A]B}
\label{DNgauge}
\ee
\end{Lemma}
{\em Proof:}
A direct calculation leads to
\be
\tilde{D}_{[C}\tilde{t}_{A]B}=D_{[C}t_{A]B}+
\frac{1}{\omega} t_{B[C}\omega_{A]}+\frac{1}{\omega}q_{B[C} t^D_{A]}\omega_D +\frac{a}{\omega} K \omega_{[C}q_{A]B} .
\label{Dtgauge}
\ee
By using the 2-dimensional identity
$$
t_{B[C}\omega_{A]}+q_{B[C} t^D_{A]}\omega_D=t^E{}_E \, q_{B[C}\omega_{A]}
$$
valid for any  symmetric tensor field $t_{AB}$, equation (\ref{Dtgauge}) can be rewritten simply as \eqref{Dtgauge2}.
$\Box$

Two important corollaries follow.
\begin{Corollary}\label{coroDt}
A symmetric tensor field $t_{AB}=t_{(AB)}$ on ${\cal S}$ whose gauge behaviour under residual gauge transformations is given by (\ref{normalgauge}) satisfies
$$
\tilde{D}_{[C}\tilde{t}_{A]B}=D_{[C}t_{A]B}
$$
if and only if its trace is $t^C{}_C =a K$.

In particular, a symmetric and gauge-invariant tensor field $\tilde{N}_{AB}=N_{AB}=N_{(AB)}$ on ${\cal S}$ satisfies
$$
\tilde{D}_{[C}\tilde{N}_{A]B}=D_{[C}N_{A]B}
$$
if and only if it is traceless $N^C{}_C =0$.
\end{Corollary}

\begin{Corollary}\label{coroRho}
If $\Sc$ has $\mathbb{S}^2$-topology, there is a unique symmetric tensor field $\rho_{AB}$ whose gauge behaviour is (\ref{eq:rho-rho'}) and satisfies the equation
\be
D_{[C}\rho_{A]B}=0\label{Drho}
\ee
in any gauge. Furthermore, this tensor field must have a trace $\rho^E{}_E = K$ ---and is given, for round spheres, by $\rho_{AB} = q_{AB} K/2$.
\end{Corollary}
{\em Proof:}
Uniqueness follows from that of trace-free Codazzi tensors on $\mathbb{S}^2$ Riemannian manifolds, by  noticing that Corollary \ref{coroDt} implies that any such $\rho_{AB}$ has a fixed trace given by $K$ and the assumption that (\ref{Drho}) holds in any gauge. Existence can be deduced directly by noticing that $\rho_{AB} = q_{AB} K/2$ is such that $D_C \rho_{AB}=0$ in the round metric sphere.
$\Box$

\noindent
Let $\vec\chi$ denote any conformal Killing vector on $(\mathbb{S}^2,q_{AB})$. Then, as proven in \cite{Fernandez-Alvarez_Senovilla-dS} the symmetric tensor field
$$
\lied_\chi \rho_{AB} + \frac{1}{2} D_A D_B D_C\chi^C
$$
is trace- and divergence-free and gauge invariant under \eqref{gaugeh}. Therefore, it must vanish on the sphere. Thus, for any conformal Killing vector on $(\mathbb{S}^2,q_{AB})$ we  have
\be\label{lierho}
\lied_\chi \rho_{AB} =-\frac{1}{2} D_A D_B D_C \chi^C .
\ee
For manifolds $\Sc$ with other topologies, if they contain a conformal Killing vector $\vec\chi$ with a fixed point --which necessarily generates an axial conformal symmetry around the fixed point \cite{Fernandez-Alvarez_Senovilla-dS,Mars1993}--, the uniqueness of $\rho_{AB}$ can also be proven by adding \eqref{lierho} for that $\vec\chi$ as an assumption. The existence of such a conformal Killing vector is ensured if the topology of $\Sc$ is either $\mathbb{S}^2$ or $\mathbb{S}^1\times\mathbb{R}$ or $\mathbb{R}^2$.

This `magic' tensor $\rho_{AB}$ allows us to derive the following non-trivial result.
\begin{Lemma}\label{lem:intLieK}
Let $(\mathbb{S}^2,q_{AB})$ be any Riemannian manifold on the 2-sphere. Then, for every conformal Killing vector field $\vec\xi$
\be\label{intLieK}
\int_{\mathbb{S}^2} \pounds_{\vec \xi} K  =0.
\ee 
\end{Lemma}
\begin{proof}
Let $(\mathbb{S}^2,q_{AB})$ be any Riemannian manifold on the 2-sphere, and let $\rho_{AB}$ be the unique tensor field on $(\mathbb{S}^2,q_{AB})$ of Corollary \ref{coroRho}. Then
$$
D_C(\rho^C{}_A -\delta^C_A K) =0
$$
and this statement is conformally invariant. Contracting here with $\xi^A$ and integrating one easily gets
$$
0=- \frac{1}{2} \int_{\mathbb{S}^2}KD_C \xi^C  = \frac{1}{2} \int_{\mathbb{S}^2} \xi^C D_C K .
$$
\end{proof}
This result seems to have been found first in \cite{Bourguignon1987}, see also references therein, and is actually valid for arbitrary {\em compact} Riemannian manifolds, also in higher dimensions if the scalar curvature is used instead of $K$. In that paper they also prove for arbitrary compact manifolds 
\begin{Lemma}\label{useful}
Let $({\cal S},q_{AB})$ be any compact 2-dimensional Riemannian manifold. Then 
$$
\int_{\mathbb{S}^2} \Delta f\,  \pounds_{\vec\xi} f =0,  \hspace{1cm} \forall f\in C^2({\cal S})
$$
and this statement is conformally invariant.
\end{Lemma}

In explicit calculations, it  is sometimes useful to have the version of \eqref{eq:rho-rho'} that provides $\rho_{AB}$ in terms of $\tilde{\rho}_{AB}$, the conformal metric metric $\tilde{q}_{AB}$ and its covariant derivative $\tilde{D}_A$, which reads
\be
\rho_{AB}=\tilde{\rho}_{AB} +\frac{1}{\omega} \tilde{D}_A\omega_B-\frac{1}{2\omega^2} \tilde{q}^{CD} \omega_C \omega_D \tilde{q}_{AB} .
\ee

If the 2-dimensional metric has axial symmetry, one can present an explicit expression of the tensor $\rho_{AB}$ in explicit adapted coordinates $\{x^A\}=\{p,\varphi\}$, with $\partial_\varphi$ the axial Killing vector. Let the metric be
$$
q_{AB}dx^A dx^B = F(p) dp^2 +G(p) d\varphi^2 
$$
where $F$ and $G$ are arbitrary functions of $p$ only subject to satisfy the necessary regularity condition at the fixed point of $\partial_\varphi$ \cite{Mars1993}. This metric is (locally) conformal to the round metric (\ref{round}), so that by adapting the coordinates on the round sphere to make the fixed point of $\partial_\varphi$ coincide with either $\theta=0$ or $\theta=\pi$ in \eqref{round}. Then, the tensor $\rho_{AB}$ is explicitly given by
\bean
\rho_{pp} &=& \frac{F}{2G} \sin^2\theta -\Psi' +\frac{F'}{2F} \Psi -\frac{1}{2} \Psi^2, \\
\rho_{p\varphi}&=&0,\\
\rho_{\varphi\varphi}&=&\frac{1}{2} \sin^2\theta +\frac{\Psi}{2F}(G\Psi-G')
\eean
where primes are derivatives with respect to $p$ and
$$
\tan \frac{\theta}{2} =b e^{\epsilon\int\sqrt{F/G}dp},  \hspace{7mm} \Psi = \frac{G'}{2G} -\epsilon \sqrt{\frac{F}{G}} \cos\theta
$$
with $\epsilon^2=1$ a sign while $b$ is a constant to be determined at the fixed point depending on the choice of $\theta=0,\pi$.

With these formulas at hand, one can easily derive that, for the flat metric with $F(p)=1$ and $G(p)=p^2$, the tensor $\rho_{AB}|_{flat}=0$ vanishes
\cite{Fernandez-Alvarez_Senovilla-dS}.

\section[\appendixname~\thesection]{Analysis of (\ref{cond}) based on the Hodge decomposition}\label{App:Hodge}
On $(\mathbb{S}^2,q_{AB})$ the Hodge theorem applies and thus any one-form $\bm{X}$ can be decomposed, uniquely, into an exact one-form, plus a co-exact one-form, plus a harmonic one-form, the latter in the cohomology class as $\bm{X}$. As $\mathbb{S}^2$ is simply connected, the harmonic one-form necessarily vanishes and thus (using $\star$ for the Hodge operator on $(\mathbb{S}^2,q_{AB})$)
$$
\bm{X} = \star d \star X_{[2]} +dX, \hspace{1cm} X_A = D^B X_{AB} +D_A X
$$
for some 2-form $X_{AB}=X_{[AB]}$ and scalar field $X$ subject to the freedom $X_{AB}\rightarrow X_{AB}+c_1\epsilon_{AB}$ and $X\rightarrow X+c_2$, with $c_1,c_2$ arbitrary constants. Notice that 
$$
X_{AB} =\epsilon_{AB} x, \hspace{1cm} x:= \star X_{[2]} = \frac{1}{2} \epsilon^{AB} X_{AB} 
$$
so that the above formula can be re-expressed in terms of two scalar fields $x$ and $X$:
\be\label{Hodge}
X_A = \epsilon_A{}^B D_B x + D_A X = D_B\left(\epsilon_A{}^B x +\delta_A{}^B X \right)
\ee
with
\be\label{Deltaxs}
\epsilon^{AB} D_A X_{B} =-\Delta x, \hspace{1cm} D_A X^A = \Delta X.
\ee
From (\ref{Hodge}) one readily obtains
\be\label{XAXA}
q^{AB}X_A X_B = D_Bx D^B x + D_B X D^B X +2 \epsilon^{AB} D_A X D_B x
\ee

The dual decomposition is simply 
$$
(\star X)_A = \epsilon_A{}^B D_B X - D_A x =D_B\left(\epsilon_A{}^B X -\delta_A{}^B x \right).
$$
Observe that $x$ and $X$ are gauge invariant if and only if $X_A$ is gauge invariant. In our case, we are rather interested in the situation where $X_A$ has gauge behaviour (\ref{Xgauge}). The relation between the $x,X$ in one gauge and $\tilde{x},\tilde{X}$ in another gauge is not trivial. 

There exists a decomposition for symmetric and traceless tensors (see e.g. \cite{Chen2021}) $H_{AB}$, analogous to (\ref{Hodge}) and also with two potentials , say $h$ and $H$, given by
\be\label{Hodge2}
H_{AB} = D_A D_B H -\frac{1}{2} q_{AB} \Delta H+ \epsilon_{(A}{}^E D_{B)} D_E h 
\ee
which also has a dual version
$$
\epsilon_{(A}{}^E H_{B)E} = \epsilon_{(A}{}^E D_{B)} D_E H -( D_A D_B h -\frac{1}{2} q_{AB} \Delta h).
$$
Notice that $H,h$ being functions on the sphere, they can be expanded in spherical harmonics as explained below for $x$ and $X$, but the harmonics with spin $s=0,1$ do not contribute to the formula (\ref{Hodge2}). In other words, the potentials $H,h$ are defined up to addition of arbitrary harmonics with $s=0,1$. 
These formulas can be applied, for instance, to $\hat{C}_{AB}$, $\hat{\F}_{AB}$ or $\S_{AB}$. 

Fortunately, the analysis of the the {\em gauge-invariant} condition (\ref{cond}) can be done in any gauge, in particular in one where the metric of the cut ${\cal S}$ is the round metric (\ref{round}). We have, for any CKV $\vec \xi$, using (\ref{Hodge})
$$
\int_{{\cal S}} \xi^A X_A  = \int_{{\cal S}} \xi^A \left(\epsilon_A{}^B D_B x +D_A X \right)=
 \int_{{\cal S}} \left(x\, \epsilon^{AB} D_A \xi_B  -D_C \xi^C X \right) .
$$
It follows from this expression that the term with $X$ is irrelevant for Killing vectors (as $D_C \xi^C=0$ then), while the term with $x$ is irrelevant for conformal Killing vectors, for we proved in Appendix \ref{App:rho}  that all of them are closed as one-forms (and thus $D_{[A}\xi_{B]}=0$ for them). Taking also into account that, for the Killing vectors (\ref{Kil1}-\ref{Kil3}), a direct calculation provides
$$
\epsilon_{AB} D^A \tilde\xi_{(i)}^B = 2n_{(i)}, \hspace{1cm} \forall i=1,2,3,
$$
it easily follows that the condition (\ref{cond}) splits into two similar relations for $x$ and $X$:
\be\label{condxX}
\int_{{\cal S}} x n_{(i)} =0, \hspace{1cm} \int_{{\cal S}} X n_{(i)} =0, \hspace{1cm} \forall i=1,2,3.
\ee
But $n_{(i)}$ are the spherical harmonics of degree $s=1$, and thus the above relations simply express that both $x$ and $X$ must be $L^2$-orthogonal to $Y_1^i$.

$x$ and $X$ being functions on $\mathbb{S}^2$, they can be expanded in spherical harmonics, that is
$$
x= \sum_{s=0}^\infty x^{i_1\dots i_s} n_{(i_1)} \dots n_{(i_s)}, \hspace{1cm} X= \sum_{s=0}^\infty X^{i_1\dots i_s} n_{(i_1)} \dots n_{(i_s)},
$$
where $x^{i_1\dots i_s}$ and $X^{i_1\dots i_s}$ are (for $s\geq 2$) fully symmetric and traceless `constant tensors' 
$$
X^{i_1\dots i_s} =X^{(i_1\dots i_s)}, \hspace{1cm} x^{i_1\dots i_s} =x* {(i_1\dots i_s)}, \hspace{1cm}
\delta_ {i_1i_2} X^{i_1\dots i_s}=\delta_{i_1i_2} x^{i_1\dots i_s}=0
$$ 
and they are totally traceless in the sense that contraction on any two indices with $\delta_{ij}$ vanishes. Therefore, condition (\ref{cond}) re-expressed as (\ref{condxX}) simply implies that the terms with $s=1$, $x^i$ and $X^i$ vanish. As $x$ and $X$ are defined up to the addition of an arbitrary constant, one can also get rid of the terms with $s=0$ and (\ref{condxX}) imply the following expansions
$$
x= \sum_{s=2}^\infty x^{i_1\dots i_s} n_{(i_1)} \dots n_{(i_s)}, \hspace{1cm} X= \sum_{s=2}^\infty X^{i_1\dots i_s} n_{(i_1)} \dots n_{(i_s)}.
$$
Introducing these expressions into (\ref{Hodge}) one gets for the solution of (\ref{cond}) 
\be
X_A=\sum_{s= 2} ^\infty s\left(x^{i_1\dots i_s} \epsilon_{AB}\eta^{B}_{(i_1)} +X^{i_1\dots i_s}\eta_{(i_1)A}\right) n_ {(i_2)} \dots n_{(i_s)} .\label{XAexpanded}
\ee

Let now $\{\bm{v},\star\bm{v}\}$ be an appropriate ON basis on $\mathbb{S}^2$ (this can be chosen to be the eigenbasis of $C_{AB}$, or of $\F_{AB}$, etcetera but those choices are not compulsory and thus $v^A$ must be seen as an arbitrary unit vector field). One can thus express all the conformal Killing vector fields in this basis, so that
$$
\eta^A_{(i)} = f_{(i)} v^A +g_{(i)} \epsilon^{AB} v_B , \hspace{1cm} (\star\eta)_{(i)}^A =\tilde\xi_{(i)}^A = -g_{(i)} v^A+f_{(i)}\epsilon^{AB} v_B .
$$

The scalar products of the conformal Killing vectors are known (or can be directly computed)
\bea
\vec{\tilde\xi}_{(i)}\cdot\vec{\tilde\xi}_{(j)}&=& \vec\eta_{(i)}\cdot \vec\eta_{(j)} =q^{AB} D_A n_{(i)} D_B n_{(j)} = \frac{1}{K} (\delta_{ij} -n_{(i)} n_{(j)}),\label{etadoteta}\\
\vec\eta_{(i)}\cdot \vec{\tilde\xi}_{(j)}&=& \frac{1}{K} \epsilon_{ij}{}^k n_{(k)}  \label{etadotxi}.
\eea
Another interesting identity is
\be
n_{(i)} \vec\eta_{(i)} =\vec 0, \hspace{1cm} n_{(i)} \vec{\tilde\xi}_{(i)} =\vec 0\label{neta}
\ee
where sum on $i$ understood.
The functions $f_{(i)}, g_{(i)}$ are thus subject, due to (\ref{etadoteta}-\ref{etadotxi}), to the following relations
$$
f_{(i)}  f_{(j)} + g_{(i)}  g_{(j)} = \frac{1}{K} \left(\delta_{ij}- n_{(i)}n_{(j)}\right), \hspace{1cm}
  f_{(j)}  g_{(i)}-  f_{(i)}  g_{(j)} =\frac{1}{K}\epsilon^{ijk} n_{(k)} 
$$
and due to (\ref{neta})
$$
\delta^{ij} n_{(i)} f_{(j)} =0, \hspace{1cm} \delta^{ij} n_{(i)} g_{(j)} =0.
$$
In simpler words, $\{n^{(i)}, f^{(i)},g^{(i)}\}$ constitute an orthonormal triad in the standard flat space. 
Using this in (\ref{XAexpanded}) one arrives at the expression
\be\label{condsol}
X_A= \sum_{s=2}^\infty s n^{(i_2)} \dots n^{(i_s)}\left[\left(X_{i_1\dots i_s} f^{(i_1)} -x_{i_1\dots i_s} g^{(i_1)} \right)v_A + \left(X_{i_1\dots i_s} g^{(i_1)} +x_{i_1\dots i_s} f^{(i_1)} \right)\epsilon_{AB} v^B\right].
\ee

%%%%%%%%%%%%%%%%%%%%%%%%%%%%%%%%%%%%%%%%%%
%\begin{adjustwidth}{-\extralength}{0cm}

%\reftitle{References}


\bibliographystyle{plain}

\bibliography{SenovillaUniverse}

\end{document}